\documentclass[12pt,a4paper]{article}
\usepackage[centertags]{amsmath}
\usepackage{amsfonts,amsthm,amssymb}
\usepackage{amssymb}
\usepackage{amsmath}
\usepackage{graphicx}
\usepackage{booktabs}
\usepackage{appendix}
\usepackage[hmargin=3.175cm,vmargin=3.150cm]{geometry}
\usepackage{tikz}
\usepackage{url}
\usepackage{graphicx}
\usepackage{accents}
\usepackage{booktabs}
\usepackage{appendix}
\usepackage{etoolbox}
\usepackage{xparse}
\usepackage{enumitem}
\usepackage{authblk}
\usepackage{color}
\usepackage{chngcntr}
\usepackage{authblk}
\usepackage{multirow}
\usepackage[authoryear]{natbib}
\usepackage{float}
\usetikzlibrary{shapes,arrows}
\newcommand{\real}{\mathbb{R}}

\newcommand{\gor}{\rightarrow}

\newcommand{\p}{\mathbf{P}}

\newcommand{\E}{\mathbb{E}}
\newcommand{\BI}{\mathcal{BO}}

\newcommand{\CI}{\mathcal{CI}}

\newtheorem{lemma}{Lemma}
\newtheorem{proposition}{Proposition}

\newtheorem{theorem}{Theorem}

\newtheorem{conjecture}[theorem]{Conjecture}

\theoremstyle{definition}
\newtheorem{definition}{Definition}

\newtheorem*{proposition*}{Proposition}
\newtheorem*{conj*}{Conjecture}

\begin{document}
\title{Robust Forecast Aggregation%
\thanks{This research is supported by Israel Science Foundation grants 2021296 and 2018889, United States-Israel Binational Science Foundation and National Science Foundation grant 2016734, the German-Israel Foundation grant I-1419-118.4/2017, the Ministry of Science and Technology grant 19400214, Technion VPR grants, and the Bernard M. Gordon Center for Systems Engineering at the Technion.}%
}
\author{Itai Arieli}
\affil{\small Faculty of Industrial Engineering and Management, Technion -- Israel Institute of Technology. E-mail: iarieli@technion.ac.il}
\author{Yakov Babichenko}
\affil{\small Faculty of Industrial Engineering and Management, Technion -- Israel Institute of Technology. E-mail: yakovbab@technion.ac.il}
\author{Rann Smorodinsky}
\affil{\small Faculty of Industrial Engineering and Management, Technion -- Israel Institute of Technology. E-mail: rann@ie.technion.ac.il}


\renewcommand\Authands{, and }
\renewcommand\footnotemark{}
\maketitle

\vspace*{-6mm}
\begin{abstract}
\noindent Bayesian experts who are exposed to different evidence often make contradictory probabilistic forecasts. An aggregator, ignorant of the underlying model, uses this to calculate her own forecast. We use the notions of {\em scoring rules} and {\em regret} to propose a natural way to evaluate an {\em aggregation scheme.} We focus on a binary state space and construct low regret aggregation schemes whenever there are only two experts which are either Blackwell-ordered or receive conditionally i.i.d. signals. In contrast, if there are many experts with conditionally i.i.d. signals, then no scheme performs (asymptotically) better than a $(0.5,0.5)$ forecast. 
\end{abstract}

\maketitle


\section{Introduction}\label{sec:intro}
Just the other day we were planning our weekend activities and looked at weather forecast for Tel Aviv on Friday, January 27. In particular, what interested us was the probability of rain. Accuweather's precipitation forecast was  77\%, Yahoo!'s was 60\%, and the Weather Channel's was 90\% (all three screenshots are provided in the Appendix). It was unclear to us how to aggregate these conflicting forecasts, although we knew that all three were from reputable sources that were using sound weather models and reliable data.

Our dilemma is not unique. In fact, many of us face conflicting advice from experts on a daily basis: forecasts from reliable pollsters on the outcome of presidential elections, medical prognoses from trusted physicians, investment advice from experienced financial pundits, and more.

This challenge is in fact inherent in many governing bodies. In the political arena we often see ministers and legislators who as elected officials must decide on critical issues and policies while lacking subject-matter expertise. These publicly elected officials dictate health care policies, decide on military development and deployment, financial regulation, and so on, without any  medical/military/financial background. As a result, they reach out to experts for advice, such as ad-hoc committees, civil servants with years of experience, lobbyists, and more. Similar to elected officials, board members of commercial companies are often seasoned  business people with managerial experience who often lack  industry-specific knowledge. These board members essentially need to aggregate input from various experts in order to make a decision.

We consider a model with two types of agents. A set of non-strategic \emph{experts} share a common prior over the state space. Each expert receives a private signal that induces a posterior distribution (the expert's forecast). In contrast, an \emph{ignorant aggregator} is not familiar with the common prior and the signal structure (we refer to this pair as the information structure). The aggregator observes the experts' forecasts and must aggregate them into a single forecast. How should we evaluate the aggregator and what is his best course of action? These are the questions we are interested in.

To study this we first elucidate four aspects of the model:
\begin{itemize}
\item How to measure the accuracy of a forecast?
A natural and prevalent family of measures of forecast accuracy, and the one we adopt here, is that of proper scoring rules (see \citep{brier}) and in particular the square loss function. The appealing property of proper scoring rules is that they induce a Bayesian expert to be truthful about his forecast.
\item
How to model an expert?
An expert has some prior distribution over the state space and receives a private signal that he then uses to compute a posterior distribution using Bayes rule. All experts share a common prior but signals are private.
\item
How to model an ignorant aggregator?
An aggregator is ignorant if his forecast is a function of a vector of experts' forecasts only. In particular, an ignorant aggregator's forecast is not a function of the underlying information structure. Obviously, his forecast is inferior compared with  some hypothetical \emph{omniscient expert}. This omniscient expert knows the vector of forecasts and in addition knows the information structure and the vector of private signals observed by the experts.
The score attained by the hypothetical omniscient expert's forecast serves as a benchmark for out ignorant aggregator.
\item
How to evaluate an ignorant aggregator's performance? For a given scoring rule and information structure one can compute the ignorant aggregator's \emph{relative loss}, which is the difference between his expected score and that of the omniscient expert. However, given that the ignorant aggregator is unfamiliar with the information structure, it is not clear whether it should serve to evaluate his performance. Here we adopt the robust (or, equivalently, the adversarial) approach. We say that the ignorant aggregator can guarantee a regret of $\alpha$ in a given class of information structures if his relative loss is at most $\alpha$, for \emph{all} information structures in the given class.
\end{itemize}

One important aspect of our model is that it pertains to a single interaction. In particular, the aggregator has no prior experience with the experts and he cannot observe past realizations.
We argue that this single interaction condition is realistic in some settings. Aggregating prognoses from physicians is typically a one-off challenge, for example. Similarly, aggregating economic forecasts is important when deciding on a mortgage and for many of us this is the only time it is called for. However, even if aggregators repeatedly interact with experts and have an opportunity to learn, there is always the challenge of the first interaction and often, typically with publicly elected officials, the first impression matters.

\subsection{Our Contribution}
Our results are as follows. We start with the case where there are two experts. We notice (see Proposition \ref{pro:g} in Section \ref{sec:general info struct}) that without any restriction on the information structure the ignorant aggregator cannot guarantee any regret below $\frac{1}{4}$. A regret of $\frac{1}{4}$ is trivially achievable by constantly announcing $\frac{1}{2}$ \emph{irrespective of experts' forecasts} (we recall that we measure accuracy by the square loss function). We proceed with the case where the two experts are Blackwell-ordered, in the sense that one is strictly better informed than the other. In this case Theorem \ref{th:m} provides an \emph{exact} formula for the minimal regret, $\frac{1}{8}(5\sqrt 5-11)\approx 0.0225$, as well as the aggregation scheme that guarantees it. We introduce an aggregation scheme that is based on the \emph{precision} of the two forecasts that guarantees a regret of $\frac{1}{8}(5\sqrt 5-11)$, and prove that no aggregation scheme can guarantee a regret below $\frac{1}{8}(5\sqrt 5-11)$.


We then study the case where the experts' signals are distributed independently conditional on the realized state. In such an environment the prior and the experts' forecasts are a sufficient statistic to perform the optimal Bayesian aggregation. In our case, the ignorant aggregator does not know the prior. A natural approach in such an environment is to guess the prior and perform the aggregation \emph{as if} the guess equaled the actual prior. This approach yields the \emph{average prior aggregation scheme}, where the guess is simply the average of the two forecasts and results in a regret of $0.0260$ (see Theorem \ref{th:iid}). This need not be the optimal scheme but we show that the optimal scheme cannot do much better and is bounded below by $\frac{1}{8}(5\sqrt 5-11)\approx 0.0225$, the exact same regret as that of the Blackwell-order setting. We discuss the gap of $0.0260 - 0.0225=0.0035$ and some related conjectures in Appendix \ref{app:ci}.

Finally, we consider the case of a large number of experts $n$ with independent and identically distributed (i.i.d.) signals.
For this case we show in Theorem \ref{th:freg} that poor performance is unavoidable and the best possible regret that can be guaranteed by the aggregator, as a function of $n$, approaches $\frac{1}{4}$ as $n\rightarrow \infty$. This result highlights the significance of the common prior assumption in large-scale aggregation of information. We note that in order to apply optimal information aggregation (i.e., the aggregation achieved by the omniscient expert) it is sufficient for the aggregtor to know the prior; optimal aggregation does not require a knowledge of the information structure or the signals, but of just the forecasts and the prior. Without the knowledge of the prior, the ignorant aggregator cannot identify the state with probability one, as has been demonstrated by \citep{PSM} and \citep{BabAriSmo2017a}. 

We show a much stronger negative result. Not only can the aggregator not identify the state, he cannot even aggregate the information into some intelligent forecast other than $\frac{1}{2}$. In other words, for the worst-case information structure, the aggregator's (approximately) best course of action is to ignore the forecasts and predict $\frac{1}{2}$, which guarantees him a loss of $\frac{1}{4}$. There is no procedure that significantly improves upon this one.

 \subsection{Related Literature}

Forecast aggregation is a strand of the statistics literature both within the classic approach and within the more modern ``machine learning'' approach. It encompasses three lines of research whose focus is different from ours.

First, within the Bayesian paradigm forecast aggregation studies the actual structure of the Bayesian aggregation  scheme for a variety of parametric information structures such as information that induces independent log-Normal posteriors for the experts \citep{SFU}, or the \emph{partial information framework} \citep{EPSU,SPU}.%
\footnote{See \citep{SFU} for an excellent review of this literature.}
We primarily depart from these papers by considering a forecast aggregator who does not share the experts' common prior or their information structure. All he knows is the actual forecasts.

Second, it considers data-driven heuristics for aggregating forecasts. Given some past data of forecasts and realizations, it studies the performance of various heuristics on the data set. This results in optimal heuristics for specific data structures (see, e.g., \citep{BD} for an excellent literature survey). Our approach and the resulting heuristics are agnostic to any available data as long as the underlying assumption on the information structure is valid.

The machine-learning community has developed techniques for integrating the advice of multiple experts, whether in the form of forecasts or in the form of proposed action (such as portfolio selection in a financial market setting). The goal of these techniques is regret minimization. In that model an ignorant aggregator (the ``machine'' in their jargon) repeatedly receives input from multiple experts and takes an action based on a vector of the experts' advice. At each stage, the aggregator is paid according to some function that depends on her action and the temporal state of nature.
The ``regret'' measures how much worse, in hindsight, the aggregator performs as compared with the best expert.
The literature provides a variety of settings and schemes for choosing actions such that the average per-stage regret goes to zero. The reader is referred to \citep{cesa} for a review on this topic. The major distinction with our work is that it considers a repeated setting whereas we study a one-shot model, about which the machine learning literature is mute.

Another related research topic is that of expert testing and in particular multiple expert testing. \citep{alnajjar} and  \citep{Feinbergtewart} ask how a policy maker should identify which of the experts is better informed. The advisee in their case does not aggregate the advice but rather chooses which of the experts to follow. In contrast with our model, the expert-testing setting does not assume a common prior among the experts. It assumes that the advisee has many past observations consisting of experts' forecasts and realizations and is once again mute about single-stage interaction.
One natural test for ranking experts, in the context of prediction in financial markets, is that of portfolio returns. \citep{sandroni} shows that indeed the better-informed expert outperforms the less-informed one in the long run.

Finally, in a companion paper \citep{BabAriSmo2017a} we use a similar model to study the conditions under which an aggregator can perfectly learn the state of the world (and in particular obtain a regret of zero) whenever there are many i.i.d. experts.

\section{Model}
Let $\Omega=\{0,1\}$ denote the binary state of nature. An {\em information structure} for $n$ experts, denoted by $(S,\p)$, consists of some $n$-dimensional signal space, $S=S_1\times\cdots\times S_n$, and a distribution $\p\in\Delta(\Omega\times S)$.
Let  $\mu=\p(\omega=1)$  denote the prior probability of the state $\omega=1$.
Expert $i$ receives a signal $s_i\in S_i$, drawn according to $\p$, and announces his \emph{forecast}, $x_i(s_i)=\p(\omega=1|s_i)$ (his conditional  probability for the state $\omega=1$).

We consider
an \emph{ignorant  aggregator} who is ignorant with respect to the information structure and observes only the vector of experts' forecasts $x(s)=(x_1(s_1),\ldots,x_n(s_n))$. An \emph{aggregation scheme} of $n$ forecasts is a function $f:[0,1]^n \to[0,1]$.
We study settings where the ignorant  aggregator may have partial knowledge about the information structure. This partial knowledge takes the form of a subset of such structures (a class of information structures). We compare the performance of the ignorant aggregator with that of the \emph{omniscient expert}, i.e., an expert who knows $\p$ and observes all the signals of all the agents. Note that this is the most competitive benchmark we can set to evaluate an aggregator's performance. For a discussion of less competitive benchmarks, see Section \ref{sec:benchmark}.

The basic building block for evaluating the performance of the aggregators is a {\em scoring rule}, which is a function $l:\Omega\times[0,1]\gor\real$. In words, it assigns a {\em loss} to any pair of realization and forecast (probability of the state $\omega=1$).  A {\em proper scoring rule} \citep{brier}  is a scoring rule for which the minimal expected loss is obtained when the forecast is equal to the actual distribution. Hence, a proper scoring rule incentivizes the omniscient expert to report the posterior probability. One prominent example of a proper scoring rule, which is central to our analysis, is the square loss function:
\begin{align*}
l(\omega,x)=\begin{cases}
(1-x)^2 \text{ if } \omega =1 \\
x^2 \ \ \ \ \ \text{ if } \omega =0.
\end{cases}
\end{align*}

Conditional on the information structure $(S,\p)$, the ignorant aggregator can only hope to do as well as the omniscient expert. The omniscient expert's best prediction is obtained using Bayes rule and is equal to $\hat{x}(s)=\p(\omega=1 |s)$, where $s=(s_1,\ldots,s_n)$. Hence the expected relative loss of the aggregation scheme $g$ is
$$L(g,\p)=
E_\p [l(\hat{x}(s),\omega)-l(f(x(s)),\omega)].$$

Given a class of information structures, $\mathcal{C}$, the regret of the aggregation scheme $g$ over $\mathcal{C}$ is the expected relative loss in the worst-case scenario:%
\footnote{The term `regret' is inspired by terminology introduced by \citep{hannan} in the context of measuring success under a worst-case scenario. In a way this is also reminiscent of \citep{gilboa}'s notion of MinMax expectation.}
\begin{equation}\label{eq:reg}
R_\mathcal{C}(f)=\sup_{\p\in \mathcal{C}}L(f,\p).
\end{equation}

We start with a preliminary observation that provides an alternative formula for the relative loss of an aggregation scheme.
\begin{lemma}\label{lem:sqdis}
For every information structure $\p$ and aggregation scheme $f:[0,1]^n\rightarrow\mathbb{R}$ it holds that $$L(f,\p)=\E_{(\omega,s_1,\ldots,s_n)\sim \p} [(f(x_1(s_1),\ldots,x(s_n))-\hat{x}(s_1,\ldots,s_n))^2].$$
\end{lemma}
\begin{proof}
For every realized vector of signals $s=(s_1,\ldots,s_n)$,
\begin{align*}
&\E_{\omega}[(f(x(s))-\omega)^2 - (\hat{x}(s))-\omega)^2 |s]=\\
& \p(\omega=1|s)[(f(x(s))-1)^2-(\hat{x}(s))-1)^2]+\\
&\p(\omega=0|s)[(f(x_1(s)))^2-(\hat{x}(s))^2]=\\
& \hat{x}(s)[(f(x(s))-1)^2-(\hat{x}(s)-1)^2]+(1-\hat{x}(s))[(f(x(s)))^2-(\hat{x}(s))^2]=
\\
&[(f(x(s))^2-2\hat{x}(s)f(x(s))+(\hat{x}(s))^2]=\\&(f(x(s))-\hat{x}(s))^2.
 \end{align*}
Since the equation holds for every $s=(s_1,\ldots,s_n)$ it holds also in expectation over $s=(s_1,\ldots,s_n)$.
\end{proof}


\section{General Information Structures}\label{sec:general info struct}
 The trivial aggregation scheme, $f(x_1,x_2)=\frac{1}{2}$, ignores the forecasts made by the two experts yet guarantees a regret of $\frac{1}{4}$. Our first observation is that no other aggregation scheme can outperform this. In fact, the following is a slightly stronger result:

\begin{proposition}\label{pro:g}
There exists an information structure $\p$, such that for every aggregation scheme $f,$ it holds that $L(f,\p)\geq \frac{1}{4}$.
\end{proposition}
\begin{proof}
Let $S_i=\{s_i,s'_i\}$ for $i=1,2$ and let $\p$ be the following distribution:
\begin{table}[H]
\centering
\begin{tabular}{cccccccc}
\multirow{4}{*}{} & \multicolumn{3}{c}{$\omega=0$}                                                    &  & \multicolumn{3}{c}{$\omega=1$}                                                    \\
                          &                             & $s_2$                    & $s'_2$                   &  &                             & $s_2$                    & $s'_2$                   \\ \cline{3-4} \cline{7-8}
                          & \multicolumn{1}{c|}{$s_1$}  & \multicolumn{1}{c|}{1/4} & \multicolumn{1}{c|}{0}   &  & \multicolumn{1}{c|}{$s_1$}  & \multicolumn{1}{c|}{0}   & \multicolumn{1}{c|}{1/4} \\ \cline{3-4} \cline{7-8}
                          & \multicolumn{1}{c|}{$s'_1$} & \multicolumn{1}{c|}{0}   & \multicolumn{1}{c|}{1/4} &  & \multicolumn{1}{c|}{$s'_1$} & \multicolumn{1}{c|}{1/4} & \multicolumn{1}{c|}{0}   \\ \cline{3-4} \cline{7-8}
\end{tabular}
\end{table}

It is easy to check that $x_i(s_i)=x_i(s'_i)=\frac{1}{2}$ and $\hat{x}(s_1,s_2)=\hat{x}(s'_1,s'_2)=0$ and $\hat{x}(s_1,s'_2)=\hat{x}(s'_1,s_2)=1$. Namely, each one of the signals separately is uninformative, but together they reveal the state of nature. The ignorant aggregator always observes two forecasts of $\frac{1}{2}$, and has no better action than forecasting $\frac{1}{2}$. On the other hand, the omniscient expert always knows the state (w.p. 1). Therefore, for every aggregation scheme $f$ the relative loss is at least $\frac{1}{4}-0$.
\end{proof}

In Section \ref{sec:aba} we show that the high regret obtained in Proposition \ref{pro:g} holds even when we benchmark the ignorant aggregator against a less challenging expert; the Bayesian aggregator.%
\footnote{The Bayesian aggregator knows the information structure and the experts' forecasts. However, in contrast with the omniscient expert, he does not observe the experts' private signals.}

Thus, ignorantly aggregating forecasts without any restriction on the family of information structures is impossible. What can be done when we consider special classes of information structures? Apparently, for some natural classes of information structures there are aggregation schemes that guarantee a surprisingly low regret.
\section{Blackwell-Ordered Experts}\label{section:blackwell}

An interesting case to analyse in our settings is the scenario where one expert is \emph{more informed} than the other.

\begin{definition}
An information structure $(S_1,S_2,\p)$ is \emph{Blackwell-ordered}, if there exists some set, $S'_2$, such that $S_2= (S_1\times S'_2)$ or, symmetrically, there exists some set, $S'_1$, such that $S_1= (S'_1\times S_2)$.  Let $\BI$ denote the set of all Blackwell-ordered information structures.
\end{definition}

In words, the better-informed expert
has access to the signal available to the less-informed expert and he receives an additional private signal.
This notion is equivalent to the notion of Blackwell domination and Blackwell ordering; see \citep{blackwell}.

Similarly to \eqref{eq:reg}, we define the regret of an aggregation scheme $f$ in Blackwell environment to be
\begin{align*}
R_\mathcal{\BI}(f)=\sup_{\p \in \BI}L(f,\p).
\end{align*}

To gain some intuition about the problem, we study the regret of two simple and naive aggregation schemes.

%

\subsection{Naive aggregation schemes}


\subsubsection*{The DeGroot scheme}

Consider the naive aggregation scheme $f(x_1,x_2)=\frac{1}{2}x_1+\frac{1}{2}x_2$, which coincides with the celebrated DeGroot opinion formation function (see \citep{degroot}). Recall that our criterion of success computes the regent under an adversarial information structure.
Consider the following information structure. Assume that the prior is $\mu=\frac{1}{2}$ and that Expert 1 receives no additional information while Expert 2 learns the realized state $\omega$. The resulting pair of forecasts will be $(\frac{1}{2},0)$ and $(\frac{1}{2},1)$, each with probability $\frac{1}{2}$. The aggregator's forecast, under the naive DeGroot scheme, will be either $\frac{1}{4}$ or $\frac{3}{4}$, each with probability $\frac{1}{2}$. In both cases the forecast will differ by $\frac{1}{4}$ from that of the better expert and hence the regret in the square loss utilities is at least $\frac{1}{16}=0.0625$.%
\footnote{In fact, this information structure leads to the worst-case relative loss and so the regret of the DeGroot scheme is exactly $\frac{1}{16}=0.0625$.}

\subsubsection*{The minimal entropy scheme}

\begin{figure}[h]
\begin{center}
\begin{tikzpicture}[xscale=0.5,yscale=0.5]

\filldraw (0,0) circle (0.2);
\filldraw (0,5) circle (0.2);
\filldraw (10,0) circle (0.2);
\filldraw (10,5) circle (0.2);
\filldraw (3,0) circle (0.2);
\filldraw (2,5) circle (0.2);
\filldraw (7,0) circle (0.2);
\filldraw (8,5) circle (0.2);
\filldraw (5,10) circle (0.2);

\draw[dashed] (0,0) -- (10,0);
\draw[dashed] (0,5) -- (10,5);
\draw[->] (5,10) -- (2.2,5.2);
\draw[->] (5,10) -- (7.8,5.2);

\draw[->] (2,5) -- (0.1,0.3);
\draw[->] (2,5) -- (6.8,0.2);

\draw[->] (8,5) -- (3.2,0.2);
\draw[->] (8,5) -- (9.9,0.3);

\draw(0,-0.1) node[below] {$0$};
\draw(3,-0.1) node[below] {$0.3$};
\draw(7,-0.1) node[below] {$0.7$};
\draw(10,-0.1) node[below] {$1$};

\draw(1.5,5) node[above] {$0.2$};
\draw(8.5,5) node[above] {$0.8$};

\draw(3.5,7.5) node[left] {$0.5$};
\draw(6.5,7.5) node[right] {$0.5$};

\draw(1.7,4) node[left] {$\frac{5}{7}$};
\draw(3.2,4) node[right] {$\frac{2}{7}$};

\draw(6.8,4) node[left] {$\frac{2}{7}$};
\draw(8.3,4) node[right] {$\frac{5}{7}$};

\draw(-0.1,0) node[left] {$X_2$};
\draw(-0.1,5) node[left] {$X_1$};
\end{tikzpicture}
\end{center}
\caption{The martingale $X_1,X_2$.}\label{fig:mart}
\end{figure}
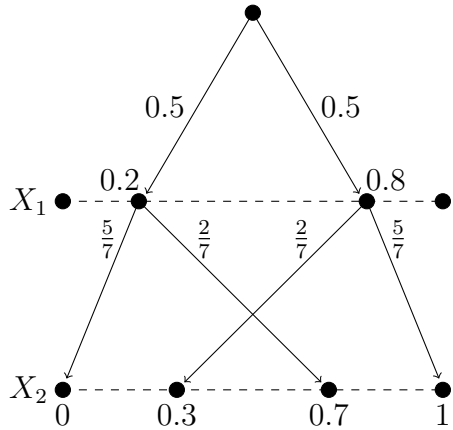

In a Bayesian framework, whenever one of the experts' forecasts is extreme ($x_i \in \{0,1\}$), he is correct with probability 1 and the aggregator should adopt his forecast independently of the information structure. A naive generalization of this is to follow the expert whose forecast is more informative, in terms of entropy. This implies adopting the more extreme forecast. Formally,
\begin{align*}
f(x_1,x_2)=\begin{cases}
x_1 \text{ if } |x_1-\frac{1}{2}|>|x_2-\frac{1}{2}| \\
x_2 \text{ otherwise.}
\end{cases}
\end{align*}
As it turns out, this aggregation scheme does not always perform well. To see this we first note that by the splitting lemma of Aumann and Maschler \citep{AM} there is an identification between Blackwell-ordered information structures and martingales $(X_0,X_1,X_2)$ of posteriors. $X_0=\mu$ is the prior, $X_1$ is the posterior of the less-informed expert, and $X_2$ is the posterior of the more-informed expert. Consider the posterior belief martingale $(X_0,X_1,X_2)$ with expectation $X_0=\frac{1}{2}$ and where $X_1=0.2,0.8$ with equal probabilities. The conditional probabilities for $X_2$ are:
$P(X_2=0|X_1=0.2) = \frac{5}{7}$,
$P(X_2=0.7|X_1=0.2) = \frac{2}{7}$ and, symmetrically,
$P(X_2=1|X_1=0.8) = \frac{5}{7}$,
$P(X_2=0.3|X_1=0.8) = \frac{2}{7}$.
Figure \ref{fig:mart} visualizes this martingale.

In this information structure with probability $\frac{1}{7}$ the ignorant aggregator's observes the pair of forecasts $(0.2,0.7)$. Based on $f$ the ignorant aggregator predicts the more extreme forecast $0.2$ whereas the omniscient expert forecast is $0.7$.
Symmetrically, with probability $\frac{1}{7}$ the ignorant aggregator observes the pair $(0.8,0.3)$ and predicts $0.8,$ which, once again, is $0.5$ away from the forecast of the better-informed expert. Thus, the induced regret is at least $\frac{2}{7}\cdot\frac{1}{4}\approx 0.0714,$ which is even worse than that of the simple average aggregation scheme.

\subsection{Optimal Aggregation}
The analysis of the two naive forecast aggregation schemes and the corresponding information structures suggests that a regret-minimizing aggregation scheme should assign weights to the forecasts  that do depend on their distance from $\frac{1}{2}$ (greater distance translates to more weight) but not too radically.
The formula of the \emph{precision scheme}, which we turn to discuss, follows this intuition. We denote by $\phi(x)=\frac{1}{x(1-x)}$ the \emph{precision} of a forecast $x$.%
\footnote{In statistics, the precision of a random variable is the reciprocal of the variance. The forecast $x$ means that expert belief is that the state is a Bernoulli random variable with probability $x$. Thus, $\phi(x)=\frac{1}{x(1-x)}$ is the precision of the forecast.}
The idea is to assign weights to the two forecasts proportional to their precision. More concretely, we define the precision scheme as follows.
\begin{align*}
f_{pre}(x_1,x_2)=
\begin{cases}
\frac{\phi(x_1)}{\phi(x_1)+\phi(x_2)}x_1 +
\frac{\phi(x_2)}{\phi(x_1)+\phi(x_2)}x_2
&\text{ if }|x_1-x_2| \leq  0.4 \\ \\
\frac{\sqrt{\phi(x_1)}}{\sqrt{\phi(x_1)}+\sqrt{\phi(x_2)}}x_1 +
\frac{\sqrt{\phi(x_2)}}{\sqrt{\phi(x_1)}+\sqrt{\phi(x_2)}}x_2
&\text{ if }|x_1-x_2| >  0.4.
\end{cases}
\end{align*}
In addition, for $x_1,x_2<1$ set $f_{pre}(0,x_2)=f_{pre}(x_1,0)=0$ and for $x_1,x_2>0$ set $f_{pre}(1,x_2)=f_{pre}(x_1,1)=1$. We also set\footnote{Note that the probability that the experts' forecasts are either $(1,0)$ or $(0,1)$, is zero in any information structure. } $f_{pre}(0,1)=f_{pre}(1,0)=\frac{1}{2}.$
Our main result for this section is the following:
\begin{theorem}\label{th:m}
For Blackwell-ordered information structures, the precision scheme guarantees a regret of $\frac{1}{8}(5\sqrt 5-11)\approx 0.0225425$. Moreover, no aggregation scheme guarantees a lower regret. In other words, $R_\BI(f_{pre})= \frac{1}{8}(5\sqrt 5-11) \leq R_\BI(f)$ for any aggregation scheme $f$.
\end{theorem}

Note that the interaction can be modelled as a zero-sum game between an ignorant aggregator (who chooses $f$) and an adversary (who chooses $\p$). The proof relies on an explicit formulation of the maxmin strategies of both the adversary and the aggregator in this zero-sum game. Once formulated, the proof is relatively easy, as it is then sufficient to show that the presented strategies guarantee the value $\frac{1}{8}(5\sqrt 5-11)$ for both sides. In Section \ref{sec:proof-dis} we provide some intuition as to how we derived these maxmin strategies.

\begin{proof}

We start by presenting an optimal strategy for the adversary. Namely, we present a distribution over two Blackwell ordered information structures, such that an aggregator who knows the mixed strategy of the adversary cannot achieve a regret below $\frac{1}{8}(5\sqrt 5-11)$. This obviously implies that our ignorant aggregator cannot achieve a better regret either.

We set the prior to $\mu=\frac{1}{2}$. The less-informed expert receives one of two signals that yield posteriors of $x\in (0,\frac{1}{2})$ and $1-x$ with equal probability $\frac{1}{2}$; i.e., the less-informed agent observes a noisy binary signal that is compatible with the correct state with probability $1-x$. Conditional on the posterior $x$, the more-informed expert observes an additional signal that yields posteriors of $0$ and $1-x$ with probabilities $\frac{1-2x}{1-x}$ and $\frac{x}{1-x}$ respectively. Such an information structure exists by the Aumann-Maschler splitting lemma \citep{AM}. Symmetrically, conditional on the posterior $1-x$, the more-informed expert observes an additional signal that yields posteriors of $x$ and $1$ with probabilities $\frac{x}{1-x}$ and $\frac{1-2x}{1-x}$ respectively. Figure \ref{fig:opt-mart} demonstrates the martingale of posteriors for the less- and more-informed experts.

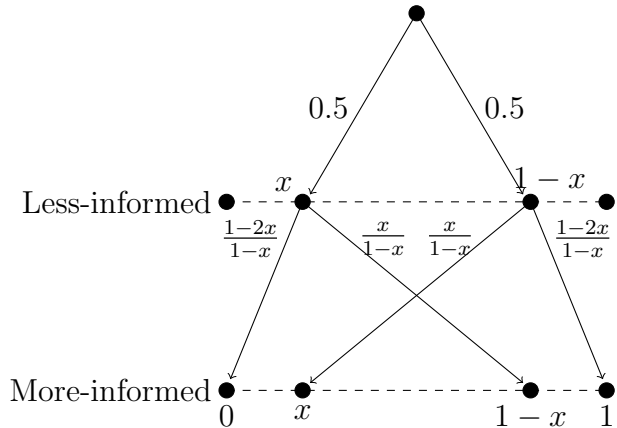
\begin{figure}[h]
\begin{center}
\begin{tikzpicture}[xscale=0.5,yscale=0.5]
\draw[dashed] (0,0) -- (10,0);
\draw[dashed] (0,5) -- (10,5);

\filldraw (0,0) circle (0.2);
\filldraw (0,5) circle (0.2);
\filldraw (10,0) circle (0.2);
\filldraw (10,5) circle (0.2);
\filldraw (2,0) circle (0.2);
\filldraw (2,5) circle (0.2);
\filldraw (8,0) circle (0.2);
\filldraw (8,5) circle (0.2);
\filldraw (5,10) circle (0.2);

\draw[->] (5,10) -- (2.2,5.2);
\draw[->] (5,10) -- (7.8,5.2);

\draw[->] (2,5) -- (0.1,0.3);
\draw[->] (2,5) -- (7.8,0.2);

\draw[->] (8,5) -- (2.2,0.2);
\draw[->] (8,5) -- (9.9,0.3);

\draw(0,-0.1) node[below] {$0$};
\draw(2,-0.1) node[below] {$x$};
\draw(8,-0.1) node[below] {$1-x$};
\draw(10,-0.1) node[below] {$1$};

\draw(1.5,5) node[above] {$x$};
\draw(8.5,5) node[above] {$1-x$};

\draw(3.5,7.5) node[left] {$0.5$};
\draw(6.5,7.5) node[right] {$0.5$};

\draw(1.7,4) node[left] {$\frac{1-2x}{1-x}$};
\draw(3.2,4) node[right] {$\frac{x}{1-x}$};

\draw(6.8,4) node[left] {$\frac{x}{1-x}$};
\draw(8.3,4) node[right] {$\frac{1-2x}{1-x}$};

\draw(-0.1,0) node[left] {More-informed};
\draw(-0.1,5) node[left] {Less-informed};
\end{tikzpicture}
\end{center}
\caption{The martingale of posteriors.}\label{fig:opt-mart}
\end{figure}

Now consider the mixed strategy where the more informed expert is chosen to be Expert 1 or Expert 2 with equal probability $\frac{1}{2}$. In the case where the ignorant aggregator observes the pair of forecasts $\{x,1-x\}$, which occurs with probability $\frac{x}{1-x}$, he does not know who the better-informed expert is. In fact he assigns equal probability $\frac{1}{2}$ to the event that Expert $i=1,2$ is the more-informed expert. Therefore, his optimal prediction in such a case is $\frac{1}{2}$, which is $(\frac{1}{2}-x)$-far from the omniscient expert forecast. Thus, the relative loss of any aggregation scheme against this mixed strategy is at least $\frac{x}{1-x}(\frac{1}{2}-x)^2$. Maximizing over $x\in (0,\frac{1}{2})$ yields a regret of $\frac{1}{8}(5\sqrt 5-11)$, which is obtained for $x=\frac{1}{4}(3-\sqrt{5})$.

We now prove that the average prior scheme guarantees a regret of at most $\frac{1}{8}(5\sqrt 5-11)$. Note that the precision scheme is anonymous, namely, $f_{pre}(x_1,x_2)=f_{pre}(x_2,x_1)$. Therefore, the adversary's best reply against the precision scheme contains an information structure where Expert 2 is the more-informed expert. Henceforth, we restrict attention to such information structures. Now, the adversary's strategy can be viewed as a martingale $X_0,X_1,X_2$ of length 2,  where $X_0=\mu$ is the prior, and $X_i$ is the posterior of Expert $i$. By Lemma \ref{lem:sqdis} the relative loss is given by $L(f_{pre},(X_0,X_1,X_2))=\E_{x_i \sim X_i} [(f_{pre}(x_1,x_2)-x_2)^2]$.

\begin{figure}[h]
\begin{center}
\begin{tikzpicture}[xscale=0.6,yscale=0.6]
\draw[dashed] (0,0) -- (10,0);
\draw[dashed] (0,5) -- (10,5);

\filldraw (0,0) circle (0.2);
\filldraw (10,0) circle (0.2);
\filldraw (5,5) circle (0.2);
\filldraw (2,0) circle (0.2);
\filldraw (7,0) circle (0.2);


\draw[->] (5,5) -- (2.1,0.2);
\draw[->] (5,5) -- (6.9,0.2);




\draw(0,-0.1) node[below] {$0$};
\draw(2,-0.1) node[below] {$x$};
\draw(7,-0.1) node[below] {$z$};
\draw(10,-0.1) node[below] {$1$};

\draw(5,5.1) node[above] {$y$};

\draw(3.5,2.5) node[left] {$\frac{z-y}{z-x}$};
\draw(6,2.5) node[right] {$\frac{y-x}{z-x}$};

%

\end{tikzpicture}
\end{center}
\caption{$M_{x,y,z}$, the extreme points of the class of martingales.}\label{fig:xyz}
\end{figure}
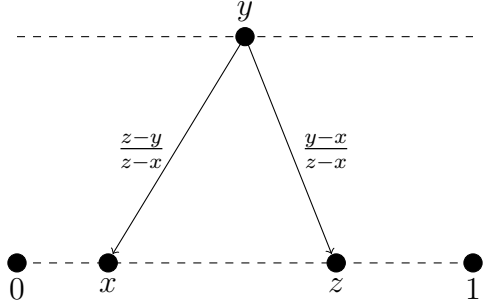

For $x<y<z$ we denote by $M_{x,y,z}$ the martingale where $X_0=X_1=y$ with probability $1$, $X_2=x$ with probability
$\frac{z-y}{z-x}$, and $X_2=z$ with probability $\frac{y-x}{z-x}$, see Figure \ref{fig:xyz}. Note that the set of martingales of length 2 is a convex set whose extreme points are exactly $\{M_{x,y,z}\}_{0\leq x \leq y \leq z\leq 1}$. Moreover, note that the adversary's utility is linear in the representation of the martingale. Namely, for a martingale that is given by the convex combination $M=\sum_{(x,y,z)\in W} \alpha_{x,y,z} M_{x,y,z}$ we have
$$L(M,f)=\sum_{(x,y,z)\in W} \alpha_{x,y,z} L(M_{x,y,z},f).$$
Therefore, given the aggregation scheme $f_{pre}$, the adversary has a best reply to $f_{pre}$ of the form $M_{x,y,z}$. From this we deduce that
$$ R_\BI(f_{pre})=\sup_{(x,y,z)\in[0,1]^3:x\leq y\leq z} L(M_{x,y,z},f_{pre}).$$
Consider the following four compact ranges in $[0,1]^3$:
\begin{eqnarray*}
&K_1=\{(x,y,z)|x\leq y\leq z,\ |x-y|\geq 0.4,\ |z-y|\geq 0.4\}\\ &K_2=\{(x,y,z)|x\leq y\leq z,\ |x-y|\leq 0.4,\ |z-y|\geq 0.4\}\\  &K_3=\{(x,y,z)|x\leq y\leq z,\ |x-y|\geq 0.4,\ |z-y|\leq 0.4\}\\
&K_4=\{(x,y,z)|x\leq y\leq z,\ |x-y|\leq 0.4,\ |z-y|\leq 0.4\}.
\end{eqnarray*}
We note that for every triplet $(x,y,z)$ where $x\leq y\leq z,$ there exists $1\leq m\leq 4$ such that $(x,y,z)\in K_m$ and the relative loss of $f$ on each of the $K_m$ is determined by a fixed loss function. For example, for $(x,y,z)$ in $K_1$ we have
\begin{align*}
&\E_{\p_{x,y,z}}[l(x_2,\omega)-l(f_{pre}(x_1,x_2),\omega)]=\\
&\frac{z-y}{z-x}\left( \frac{\sqrt{x(1-x)}y+\sqrt{y(1-y)}x}{\sqrt{x(1-x)}+\sqrt{y(1-y)}}-x\right)^2+\frac{y-x}{z-x}\left( \frac{\sqrt{y(1-y)}z+\sqrt{z(1-z)}y}{\sqrt{y(1-y)}+\sqrt{z(1-z)}}-z\right)^2.
\end{align*}
Similarly, in range $K_2$ we have
\begin{align*}
&\E_{\p_{x,y,z}}[l(x_2,\omega)-l(f_{pre}(x_1,x_2),\omega)]=\\
&\frac{z-y}{z-x}\left( \frac{x(1-x)y+y(1-y)x}{x(1-x)+y(1-y)}-x\right)^2+\frac{y-x}{z-x}\left( \frac{\sqrt{y(1-y)}z+\sqrt{z(1-z)}y}{\sqrt{y(1-y)}+\sqrt{z(1-z)}}-z\right)^2.
\end{align*}
Similar expressions can be obtained for ranges $K_4$ and $K_3$.

Thus, in order to show that $f_{pre}$ guarantees a regret of at most $\frac{1}{8}(5\sqrt{5}-11)$ to the ignorant aggregator, one need only solve four three-dimensional optimization problems in the four compact domains $\{K_i\}_{i=1,2,3,4}.$ These optimization problems has been solved by Matlab,
%
%
which shows that the global maximum of $R(x,y,z)$ is obtained at two points $(0,\frac{1}{4}(3-\sqrt{5}),1-\frac{1}{4}(3-\sqrt{5}))$ and $(\frac{1}{4}(3-\sqrt{5}),1-\frac{1}{4}(3-\sqrt{5}),1)$ and is equal to $\frac{1}{8}(5\sqrt{5}-11)$.
\end{proof}

\section{Two Conditionally Independent Experts}\label{section:ciis}

Another family of information structures that is prevalent in the economics literature is that of experts who receive independent signals, conditional on the realized state. We refer to this as a conditionally independent information structure.
In such settings, knowing the prior $\mu=\p(\omega=1)$ together with the posteriors $x_1,x_2,\ldots,x_n$ is a sufficient statistic for the omniscient expert. It is straightforward to verify the following (see, e.g., \citep{Bordley}): if  $\p(\omega=1|s_i)=x_i$ for $1\leq i\leq n,$ then
\begin{align*}
\p(\omega=1| s_1,\ldots,s_n)=g(\mu,x_1,\ldots,x_n)=\frac{(1-\mu)^{n-1}\prod_{i=1}^n x_i }{(1-\mu)^{n-1}\prod_{i=1}^n x_i  +\mu^{n-1} \prod_{i=1}^n (1-x_i)}.
\end{align*}

The ignorant aggregator, however, does not know the prior. Nevertheless, this observation induces a natural family of aggregation schemes, where Bayes rule is applied to a ``dummy'' prior (or a guess). The guess of the prior can be based on the experts' forecasts. It turns out that the resulting regret is surprisingly low when the number of experts is $2$ and the aggregator uses the standard average of the two forecasts as his guess for the prior. This guess entails the following aggregation scheme, which we refer to as the \emph{average-prior} scheme:
\begin{align*}
f_{avg}(x_1,x_2)=g(\frac{x_1+x_2}{2},x_1,x_2) =\frac{x_1 x_2 (1-\frac{x_1+x_2}{2})}{x_1 x_2 (1-\frac{x_1+x_2}{2}) + (1-x_1)(1-x_2)\frac{x_1+x_2}{2}}
\end{align*}
Let $\CI$ be the class of all independent information structures for two experts.
\begin{theorem}\label{th:iid}
$R_\CI(f_{avg})=0.0260$, that is, the average prior scheme guarantees a regret of $0.0260$. Moreover, for every aggregation scheme $f$ it holds that $R_\CI(f)\geq \frac{1}{8}(5\sqrt 5-11)\approx 0.0225$. That is, no aggregation scheme guarantees a regret lower than $\frac{1}{8}(5\sqrt 5-11)\approx 0.0225$.
\end{theorem}

\begin{proof}
We begin by  introducing a strategy for the adversary that guarantees him a regret of at least $\frac{1}{8}(5\sqrt 5-11)$. Namely, we present a distribution over two conditionally independent information structures, such that an aggregator who knows the mixed strategy of the adversary cannot achieve a regret below $\frac{1}{8}(5\sqrt 5-11)$.%
\footnote{In fact the distribution is over conditionally \emph{i.i.d.} information structures.}
This obviously implies that our ignorant aggregator cannot achieve a better regret, either.

\begin{figure}[h]
\begin{center}
\begin{tikzpicture}[xscale=0.5,yscale=0.5]
\draw[dashed] (0,0) -- (10,0);
\draw[dashed] (0,5) -- (10,5);

\filldraw (0,0) circle (0.2);
\filldraw (0,5) circle (0.2);
\filldraw (10,0) circle (0.2);
\filldraw (10,5) circle (0.2);
\filldraw (5,0) circle (0.2);
\filldraw (2,5) circle (0.2);
\filldraw (8,5) circle (0.2);
\draw (5,10) circle (0.2);

\draw[->] (5,10) -- (2.2,5.2);
\draw[->] (5,10) -- (7.8,5.2);

\draw[->] (2,5) -- (0.1,0.3);
\draw[->] (2,5) -- (4.8,0.2);

\draw[->] (8,5) -- (5.2,0.2);
\draw[->] (8,5) -- (9.9,0.3);

\draw(0,-0.1) node[below] {$0$};
\draw(5,-0.1) node[below] {$1/2$};
\draw(10,-0.1) node[below] {$1$};

\draw(1.5,5) node[above] {$x$};
\draw(8.5,5) node[above] {$1-x$};

\draw(3.5,7.5) node[left] {$0.5$};
\draw(6.5,7.5) node[right] {$0.5$};

\draw(1.3,3) node[left] {$1-2x$};
\draw(3.2,3) node[right] {$2x$};

\draw(6.8,3) node[left] {$2x$};
\draw(8.7,3) node[right] {$1-2x$};

\draw(-0.1,0) node[left] {Posterior of single expert};
\draw(-0.1,5) node[left] {Prior};
\end{tikzpicture}
\end{center}
\caption{Distribution over priors and posteriors of a single agent.}\label{fig:iid}
\end{figure}
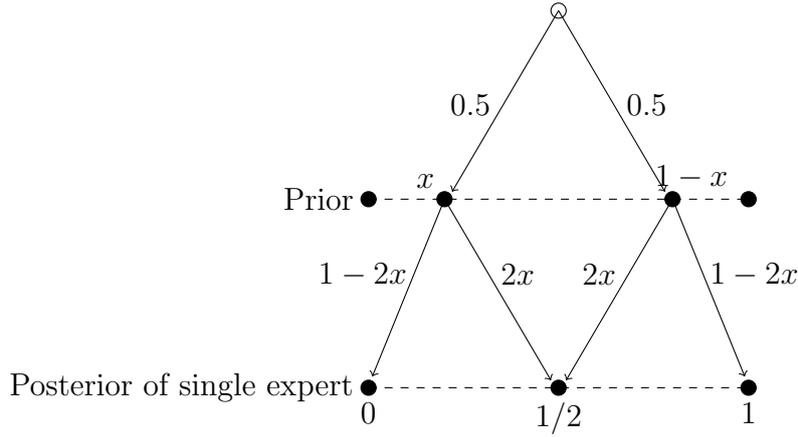

Consider the mixed strategy that randomizes uniformly over the following two conditionally independent structures. In the first information structure the prior is $x\in (0,\frac{1}{2})$. The signals are conditionally identically distributed signals that induce their posterior belief $0$ and $\frac{1}{2}$ with probabilities $1-2x$ and $2x$ respectively (by Aumann and Maschler's splitting lemma \citep{AM} such a signal structure exists). In the second information structure the prior is $1-x\in (\frac{1}{2},1)$. The posterior beliefs are $1$ and $\frac{1}{2}$ with probabilities $1-2x$ and $2x$ respectively. A visualization of the adversary's mixed strategy appears in Figure \ref{fig:iid}.

When the realizations of the two forecasts turn out to be $x_1=x_2=\frac{1}{2}$, an ignorant aggregator does not know whether the prior was $x$ or $1-x$, while a omniscient expert does. Simple symmetry arguments show that the best forecast for a ignorant aggregator in such an event is $\frac{1}{2}$, whereas a omniscient expert forecasts $$\frac{\frac{1}{2}\frac{1}{2}(1-x)}{\frac{1}{2}\frac{1}{2}(1-x)+\frac{1}{2}\frac{1}{2}x} = 1-x \text{ and } \frac{\frac{1}{2}\frac{1}{2}x}{\frac{1}{2}\frac{1}{2}x+\frac{1}{2}\frac{1}{2}(1-x)} = x,$$ depending on whether the prior was $x$ or $1-x$ respectively. By Lemma \ref{lem:sqdis} the relative loss in the event of $x_1=x_2=\frac{1}{2}$ is $(\frac{1}{2}-x)^2$.
A simple calculation shows that the probability of the event $x_1=x_2=\frac{1}{2}$ is $\frac{x}{1-x}$. Therefore, the relative loss of an aggregator who knows the adversary's strategy is
$$\frac{x}{1-x}(\frac{1}{2}-x)^2.$$
 Maximizing the relative loss over $x\in (0,\frac{1}{2})$ yields a regret $\frac{1}{8}(5\sqrt 5-11)$, which is obtained at $x=\frac{1}{4}(3-\sqrt{5})$.
This proves that no aggregation scheme can guarantee a regret below $\frac{1}{8}(5\sqrt 5-11)$.

Next we shall prove that any strategy of the adversary against the average prior scheme yields a relative loss of at most $0.0260$ to the aggregator. An adversary can be viewed as a triple $(\mu,\p^1,\p^2)$ where $\mu$ is the prior, and $\p^i \in \Delta([0,1])$ is a distribution of Expert $i$'s posterior beliefs with expectation $\E(\p_i)=\mu$. Assume that $\p_i$ assigns a prior probability of $p_i$ to the posterior $x_i$. Simple calculations show that the probability of the pair of posteriors being $(x_1,x_2)$ is given by the following  expression that is multilinear in $p_1$ and $p_2$:
$$h(p_1,p_2,\mu,x_1,x_2)= p_1 p_2 \left( \frac{(1-x_1)(1-x_2)}{1-\mu}+\frac{x_1 x_2}{\mu} \right).$$
By Lemma \ref{lem:sqdis} the relative loss of the aggregator in the case where $(x_1,x_2)$ is the realized posterior probability is
\begin{align}\label{eq:r}
\begin{split}
r(\mu,x_1,x_2)=
&\left( \frac{x_1 x_2 (1-\mu)}{x_1 x_2 (1-\mu) + (1-x_1)(1-x_2)\mu} \right.\\
&-
\left.\frac{x_1 x_2 (1-\frac{x_1+x_2}{2})}{x_1 x_2 (1-\frac{x_1+x_2}{2}) + (1-x_1)(1-x_2)\frac{x_1+x_2}{2}} \right)^2.
\end{split}
\end{align}
For every $y\leq \mu \leq z$ let $\p^i_{\mu,y,z}$ be a posterior distribution with support $\{y,z\}$ such that $x$ is realized with probability $\frac{z-\mu}{z-y}$ and $y$ is realized with probability $\frac{\mu-y}{z-y}$. These are the extreme points of the convex set of all posterior distributions. Namely,
any distribution $\p^i$ can be written as $\p^i=\sum_j \alpha^i_j \p^i_{\mu,y^i_j,z^i_j}$, where $\sum_j \alpha^1_j = \sum_j \alpha^2_j = 1$. The multilinearity of the relative loss implies the following formula on the expected relative loss:
\begin{align}\label{eq:lina}
L(f_{avg},(\mu,\p^1,\p^2))=\sum_{j,k} \alpha^1_j \alpha^2_k L((\mu,\p^1_{\mu,y^1_j,z^1_j},\p^2_{\mu,y^2_k,z^2_k}),f_{avg})
\end{align}
which is a convex combination of relative losses of the form $L(f_{avg},(\mu,\p^1_{\mu,y^1,z^1},\p^2_{\mu,y^2,z^2}))$. Therefore, the regret that an adversary can achieve against the average aggregation scheme is given by
$$\max_{\mu,y_1 \leq \mu \leq z_1, y_2 \leq \mu \leq z_2}  L((\mu,\p^1_{\mu,y^1,z^1},\p^2_{\mu,y^2,z^2}),f_{avg}).$$
Note also that the objective function has a closed formula:
\begin{align}\label{eq:5dObj}
\begin{split}
L((\mu,\p^1_{\mu,y^1,z^1},\p^2_{\mu,y^2,z^2}),f_{avg}) =&
h(\frac{z_1-\mu}{z_1-y_1},\frac{z_2-\mu}{z_2-y_2},\mu,y_1,y_2) \cdot r(\mu,y_1,y_2) + \\
& h(\frac{z_1-\mu}{z_1-y_1},\frac{\mu-y_2}{z_2-y_2},\mu,y_1,z_2) \cdot r(\mu,y_1,z_2) + \\
& h(\frac{\mu-y_1}{z_1-y_1},\frac{z_2-\mu}{z_2-y_2},\mu,z_1,y_2) \cdot r(\mu,z_1,y_2) + \\
& h(\frac{\mu-y_1}{z_1-y_1},\frac{\mu-y_2}{z_2-y_2},\mu,z_1,z_2) \cdot r(\mu,z_1,z_2).
\end{split}
\end{align}
In summary, we have shown that, given the aggregation scheme $f$, the adversary's best-reply problem can be reduced to a concrete maximization problem over five parameters $\mu,y_1,y_2,z_1,z_2$.
Matlab calculations show that the global maximum of this five-variable fraction is obtained at the point $\mu=0.120$, $y_1=0.120$, $z_1=0.120$, $y_2=0$, $z_2=0.746$ and is equal to  $0.0260$.%
\footnote{Note that the adversarial best-reply corresponds to an information structure where the first expert receives no information.}
\end{proof}

\subsection{Mind the gap}

Theorem \ref{th:iid} leaves a gap between the upper and lower regret bounds whenever the information structure is conditionally independent. We are not able to close this gap; however, we can slightly improve the upper bound by using a non-intuitive variant of the average prior scheme. That is, instead of updating the two posterior with respect to their average we determine the ``dummy'' prior as follows:
\begin{eqnarray}
ep(x_1,x_2)=
\begin{cases}
0.49 x_1 + 0.49 x_2 &\text{ if } x_1+x_2 \leq 1 \\
0.49 x_1 + 0.49 x_2 + 0.02 &\text{ otherwise.}
\end{cases}
\end{eqnarray}

\begin{proposition}\label{pro:iid}
For conditionally independent information structures, the aggregation scheme, $f(x_1,x_2)=g(ep(x_1,x_2),x_1,x_2)$, guarantees a regret of $0.0250$.
\end{proposition}

The proof of Proposition \ref{pro:iid}, which bares a similarity to the proof of the first part of Theorem \ref{th:iid}, is relegated to appendix \ref{sec:abag-proof}.

In many economic models, signals, in addition to being conditionally independent, are also {\em identical}. For this case we conjecture that the regret bound of  $\frac{1}{8}(5\sqrt 5-11)\approx 0.0225$  is tight and that the average prior scheme is indeed optimal.

\begin{conjecture}\label{con:iid}
For conditionally i.i.d. information structures, the minimal regret that can be guaranteed is equal to $\frac{1}{8}(5\sqrt 5-11)\approx 0.0225$, and the average prior scheme guarantees this regret.
\end{conjecture}
We discuss this conjecture further in Appendix \ref{app:ci}.

\section{Many Conditionally Independent Experts}
We turn to study how the regret of the ignorant aggregator is affected as the number of independent experts grows. Intuitively, the more forecasts he has at his disposal, the better the ignorant aggregator will perform. This intuition is stronger when experts are additionally assumed {\em identical}. However, one should note that the prevalence of more experts allows the omniscient expert, who serves as a benchmark, to improve as well. Hence, in terms of regret the aforementioned intuition may be misleading.

Let $\mathcal{D}_n$ be the class of all i.i.d. information structures with $n$ experts. The interplay between the improvement of the ignorant aggregator and that of the omniscient expert is given in the following theorem.
\begin{theorem}\label{th:freg}
For any number of conditionally independent and identical experts, $n$, and \emph{any} aggregation scheme, $f:[0,1]^n\rightarrow\mathbb{R}$, the following bound on regret holds: $\ R_{\mathcal{D}_n}(f)\geq \frac{1}{4}-3\sqrt{\frac{\log n}{n}}$.
\end{theorem}
Thus, as the number of agents grows no aggregation scheme can guarantee a performances that does better than the fixed scheme that always forecasts $0.5$. This stands in sharp contrast to the $n=2$ case. The proof of Theorem \ref{th:freg} is relegated to Appendix \ref{ap:reg-proof}. Below we discuss the key ideas of the proof.

\subsection{Idea of the proof of Theorem \ref{th:freg}}
\label{sec_idea_of_proof}
We start the discussion with the following example (which appears also in \citep{BabAriSmo2017a} and in \citep{PSM}).
Consider the following two information structures for a single expert:
\begin{table}[h]
\centering
\caption{The information structures $I(1),I(2)$.}
\label{tb:infor2}
$I(1):$
\begin{tabular}{ccc}
                           & $\omega=0$                          & $\omega=1$                           \\ \cline{2-3}
\multicolumn{1}{c|}{$s_0$} & \multicolumn{1}{c|}{$\frac{3}{8}$} & \multicolumn{1}{c|}{$\frac{1}{8}$}  \\ \cline{2-3}
\multicolumn{1}{c|}{$s_1$} & \multicolumn{1}{c|}{$\frac{1}{8}$} & \multicolumn{1}{c|}{$\frac{3}{8}$} \\ \cline{2-3}
\end{tabular}
\hspace*{4mm}
$I(2):$
\begin{tabular}{ccc}
                           & $\omega=0$                          & $\omega=1$                           \\ \cline{2-3}
\multicolumn{1}{c|}{$s_0$} & \multicolumn{1}{c|}{$\frac{3}{40}$} & \multicolumn{1}{c|}{$\frac{1}{40}$}  \\ \cline{2-3}
\multicolumn{1}{c|}{$s_1$} & \multicolumn{1}{c|}{$\frac{9}{40}$} & \multicolumn{1}{c|}{$\frac{27}{40}$} \\ \cline{2-3}
\end{tabular}
\end{table}
The priors in these two information structures are $\mu_1=\frac{1}{2}$ and $\mu_2=\frac{7}{10}$, respectively.

Let us denote by $(e.m.\omega)$, where $m=1,2$ and $\omega=0,1$, the event that the information structure is $I(m)$ and the state is $\omega$. A straightforward computation shows that the distribution over the experts' posteriors (forecasts) is \emph{identical}
for the two events $(e.1.1)$ and $(e.2.0)$. In both cases every expert will either observe $s_0$ with probability $\frac{1}{4}$ and forecast $\frac{1}{4}$ or will observe $s_1$ with probability $\frac{3}{4}$ and will forecast $\frac{3}{4}$. Let us denote this posterior distribution over forecasts by $\psi$

Consider an adversary's mixed strategy over these two information structures with weights $\frac{1-\mu_2}{1-\mu_2+\mu_1}=\frac{3}{8}$ assigned to information structure $I(1)$ and $\frac{\mu_1}{1-\mu_2+\mu_1}=\frac{5}{8}$ to information structure $I(2)$.
Our ignorant aggregator cannot perform better than a Bayesian aggregator, who knows the two information structures, the mixed strategy, and the forecasts (notice that, unlike the omniscient expert, this hypothetical Bayesian aggregator cannot observe the actual signals).

Assume that $n$ is large enough so that the empirical distribution of the sample of $n$ forecasts is ``essentially'' equal to the \emph{precise} posterior distribution. Thus, in events $(e.1.1)$ and $(e.2.0)$ the Bayesian aggregator will observe $\psi$ and his Bayesian belief about the event $\omega=1$ will be $\frac{3/8 \cdot \mu_1}{3/8 \cdot \mu_1 + 5/8 \cdot(1-\mu_2)}=\frac{1}{2}$. In other words, the adversary can set the probabilities over $I(1)$ and $I(2)$ such that even the Bayesian aggregator who observes the distribution over posteriors precisely will have complete uncertainty about the state in cases $(e.1.1)$ and $(e.2.0)$. All the more so is this true for the ignorant aggregator.

Unfortunately for the adversary, there are two additional cases $(e.1.0)$ and $(e.2.1)$,  where the Bayesian aggregator succeeds in determining the state. However, if we now introduce a third information structure, $I(3)$, such that the posteriors in the events  $(e.2.1)$ and $(e.3.0)$ coincide, then a mixed strategy over the three information structures can be constructed such that the same uncertainty will prevail when this posterior is observed.

For concreteness, the new information structure, $I(3)$, is
\begin{table}[h]
\label{tb:infor}
\begin{center}
$I(3):$
\begin{tabular}{ccc}
                           & $\omega=0$                          & $\omega=1$                           \\ \cline{2-3}
\multicolumn{1}{c|}{$s_0$} & \multicolumn{1}{c|}{$\frac{3}{328}$} & \multicolumn{1}{c|}{$\frac{1}{328}$}  \\ \cline{2-3}
\multicolumn{1}{c|}{$s_1$} & \multicolumn{1}{c|}{$\frac{81}{328}$} & \multicolumn{1}{c|}{$\frac{243}{328}$} \\ \cline{2-3}
\end{tabular}
\end{center}
\end{table}
and the probabilities over $I(1),I(2),I(3)$ are $\frac{9}{65},\frac{15}{65},\frac{41}{65}$.

We proceed to define the information structures $I(m)$ and the corresponding mixed strategy iteratively such that the Bayesian aggregator faces complete uncertainly unless the events are $(e.1.0)$ and $(e.m.1)$ are realized. In this construction, the posterior of all the experts in all the information structures is either $\frac{1}{4}$ or $\frac{3}{4}$.

It turns out that in order for the probability of the events $(e.1.0)\cup (e.m.1)$ to vanish, we should repeat this construction with the pair of forecasts $\frac{1}{2}\pm \epsilon$ (for sufficiently small $\epsilon>0$) instead of  $\{\frac{1}{4},\frac{3}{4}\}$.

On the other hand, in order for the omniscient expert to perform well by observing a \emph{sample} from the distribution the value of $\epsilon$ cannot be too small (note that if we set $\epsilon=0$ all the experts' forecasts will be equal to $\frac{1}{2}$, in which case the omniscient expert is also left clueless). Some tedious (yet standard) calculations  show that if we set $\epsilon=\Theta(\sqrt{\frac{\log n}{n}})$ and let the mixed strategy for the adversary be with support of size $k=\Theta(\sqrt{\frac{n}{\log(n)}})$, then we have the following two phenomena:
\begin{itemize}
\item The events $(e.1.0)\cup (e.k.1)$ occur with small probability (and thus the Bayesian aggregator cannot perform well).
\item The omniscient expert can determine the state with high probability, because the conditional distributions over posteriors $I_0(m)$ (i.e., $I(m)$ conditional on $\omega=0$) and $I_1(m)$ are sufficiently ``far" from one another, for all $1\leq m \leq k$.
\end{itemize}
\section{Discussion}

We study non-Bayesian forecast aggregation and introduce an evaluation criterion for such aggregation schemes. This criterion is based on the notion of regret with respect to a square loss utility function. There are various degrees of freedom in the model as well as in choice of an evaluation criterion which provide room for further research.

\subsection{Choice of benchmark}\label{sec:benchmark}
In this paper we study the regret of aggregation, endowed to ignorant aggregators, while  using an omniscient expert as our benchmark.  Recall that the omniscient expert knows the information structure and also observes the actual signals experts received. This benchmark, one may argue, is too challenging. A less challenging benchmark should improve the aggregator's performance.  We discuss two such alternatives.

\subsubsection{The Bayesian aggregator}\label{sec:aba}
Recall the notion of a Bayesian aggregator introduced in section \ref{sec_idea_of_proof}. This Bayesian aggregator knows the information structure and \emph{the forecasts of the experts} but does not observe the experts' private signals. Seemingly the regret associated with a Bayesian aggregator is smaller than that associated with the omniscient expert.
Apparently, for the families of information structures studied here (Theorems \ref{th:m} and Theorem \ref{th:iid}), this is not true. This follows from the fact that for such information structures both benchmarks entail a similar loss when the adversary uses his Maxmin mixed strategy.

The disappointing regret obtained in Proposition \ref{pro:g}, when the adversary has no restriction on the choice of a strategy and can use correlated signals, cannot be improved when using the Bayesian aggregator as a benchmark. This, however, is not a straightforward observation. In fact, for information structure designed in the proof of the proposition, the omniscient expert does remarkably well while the Bayesian aggregator fails miserably. Nevertheless,  a mild modification of said information structure provides the desired outcome, where the ignorant aggregator suffer a square loss of $\frac{1}{4}-\epsilon$ (where $\epsilon>0$ is arbitrarily small) while the Bayesian aggregator suffers a square loss of zero.

Below we formalize the notion of a Bayesian aggregator and state a proposition analogous to Proposition \ref{pro:g}.

A Bayesian aggregator's posterior is
$$\tilde{x}(x_1,x_2)=\p(\omega=1 | x_1(s_1)=x_1,x_2(s_2)=x_2).$$
The corresponding regret is given by
\begin{align*}
R^{AO}(f,\p)=\E_{(\omega,s_1,s_2)\sim \p} [(f(x_1(s_1),x_2(s_2))-\omega)^2] - \E_{(\omega,s_1,s_2)\sim \p} [(\tilde{x}(x_1(s_1),x_2(s_2)))-\omega)^2].
\end{align*}
\begin{proposition}\label{pro:abag}
For general information structures, there is no aggregation scheme that guarantees a regret below $\frac{1}{4}$.
\end{proposition}
The proof is relegated to Appendix \ref{sec:abag-proof}.

\noindent
\subsubsection{Best expert benchmark}
Inspired by the regret-minimization literature in \emph{repeated} expert advice settings \citep{cesa}, one may compare the square loss of the ignorant aggregator with the square loss of the better expert, namely,
\begin{align*}
R^B(f,\p)=\E_{(\omega,s_1,s_2)\sim \p} [(f(x_1(s_1),x_2(s_2))-\omega)^2] - \min_{i=1,2} \E_{(\omega,s_i)\sim \p} [(x_i(s_i)-\omega)^2].
\end{align*}
The result on Blackwell-ordered information structures remains the same, because the best expert coincides with the omniscient expert in this environment.
Interestingly, this also holds for conditionally independent information structures with two experts, where the  minimal regret is between $0.0225$ and $0.0260$.%
\footnote{We omit the underlying reasoning that leads to this result.}

For conditionally i.i.d. information structures, on the other hand, the ignorant aggregator obviously can perform as good as the best expert (in expectation) simply by mimicking Expert 1. By the symmetry of the problem, all experts suffer the same expected loss. This simple observation holds for any number of agents.

The minimal regret with respect to the best-expert benchmark in general information structures, allowing for correlation, remains an interesting open problem (even for the case of two experts).


\subsection{Extensions}
Our  model studies a simple setting where
there are only two states of nature, the scoring rule is set to be the square loss function, and we mostly focus on the case of two experts. For conditional i.i.d. signals we provided a lower bound on the regret that the aggregator can guarantee as a function of the number of experts $n$.
The simplicity of the problem was crucial for our ability to crack it. Thus, extending our results to a larger state space or other scoring rules is by no means straightforward.
We hope to further understand this in future work.

\bibliography{bibnew}
\newpage
\appendix
\section{Some intuition for the precision scheme}\label{sec:proof-dis}

\begin{figure}[h]
\begin{center}
\begin{tikzpicture}[xscale=0.6,yscale=0.6]
\draw[dashed] (0,0) -- (10,0);
\draw[dashed] (0,5) -- (10,5);

\filldraw (0,0) circle (0.2);
\filldraw (0,5) circle (0.2);
\filldraw (10,0) circle (0.2);
\filldraw (10,5) circle (0.2);
\filldraw (3,0) circle (0.2);
\filldraw (3,5) circle (0.2);
\filldraw (8,0) circle (0.2);
\filldraw (8,5) circle (0.2);
\filldraw (5,10) circle (0.2);

\draw[->] (5,10) -- (3.1,5.2);
\draw[->] (5,10) -- (7.8,5.2);

\draw[->] (3,5) -- (0.1,0.3);
\draw[->] (3,5) -- (7.8,0.2);

\draw[->] (8,5) -- (3.2,0.2);
\draw[->] (8,5) -- (9.9,0.3);

\draw(0,-0.1) node[below] {$0$};
\draw(3,-0.1) node[below] {$y$};
\draw(8,-0.1) node[below] {$x$};
\draw(10,-0.1) node[below] {$1$};

\draw(2.5,5) node[above] {$y$};
\draw(8.5,5) node[above] {$x$};

\draw(4,7.5) node[left] {$1-\alpha$};
\draw(6.5,7.5) node[right] {$\alpha$};

\draw(2.4,4) node[left] {$\frac{x-y}{x}$};
\draw(4.1,4) node[right] {$\frac{y}{x}$};

\draw(6.9,4) node[left] {$\frac{1-x}{1-y}$};
\draw(8.3,4) node[right] {$\frac{x-y}{1-y}$};

\draw(-0.1,0) node[left] {$X_2$};
\draw(-0.1,5) node[left] {$X_1$};
\end{tikzpicture}
\end{center}
\caption{The martingale $\p_{x,y,\alpha}$.}\label{fig:good2}
\end{figure}
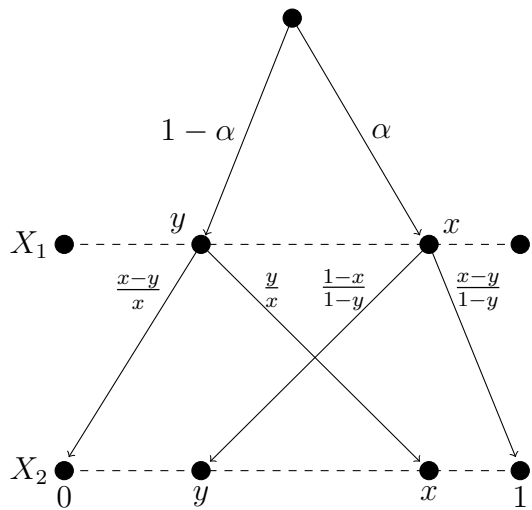
The zero-sum game that we study does not admit a pure maxmin strategy for the adversary. However, it turns out that it always admits a simple maxmin strategy for the adversary where the randomization is done only on two information structures. The adversary chooses (purely) a martingale of posteriors $X_0,X_1,X_2$ of length 2, where $X_1$ is the posterior of the less-informed expert and $X_2$ is the posterior of the more-informed expert. The randomization is done only on the identity of the less- and more-informed experts. Each expert is more-informed with probability $\frac{1}{2}$.
A formal methodology to capture this statement is to consider an \emph{equivalent} game $\Gamma$, where the adversary chooses a martingale of length 2, and the ignorant aggregator is restricted to choosing an \emph{anonymous} aggregation scheme $f(x_1,x_2)=f(x_2,x_1)$. The game $\Gamma$ is equivalent to the original game in the following respect:
\begin{itemize}
\item The value of these two games is equal.
\item A maxmin strategy of the ignorant aggregator in $\Gamma$ is a maxmin strategy in the original game.
\item A maxmin martingale of the adversary in $\Gamma$ can be translated to a maxmin strategy in the original game by choosing the less and more informed players with equal probability.
\end{itemize}

The existence of pure approximate optimal strategies in the game $\Gamma$ is a corollary of Sion's minmax theorem \citep{Sion}. Once this is verified we may focus on studying the adversary's minmax pure strategies. Note that the simplest pure strategies available to the adversary that achieve positive regret against a best-replying aggregator are martingales of the form demonstrated in Figure \ref{fig:good2}. Namely, martingales for which $X_1$ takes on two values $x,y$ with probabilities $\alpha,1-\alpha$ and $X_2$ can take only two pairs of values, depending on the outcome of $X_1$. These pairs can be either $(x,1)$ or $(0,y)$. The corresponding probabilities are uniquely determined by the martingale condition. Let us denote this family of martingales by $\mathcal{M}$, where each $M \in \mathcal{M}$ is a function of three parameters, $M = M_{x,y,\alpha}$.

The importance of the family $\mathcal{M}$ is due to the following observation:

\begin{proposition}\label{prop_3}
The adversary's minmax strategy is in $\mathcal{M}$.
\end{proposition}

In hindsight, Proposition \ref{prop_3} obviously follows from Theorem \ref{th:m}, whose proof, in turn, does not rely on this proposition. However, our original proof of Proposition \ref{prop_3} was independent and did not use the specific structure of the optimal aggregation scheme. We omit the proof, which requires tedious computations.

If the ignorant aggregator knows the martingale $M_{x,y,\alpha}$, then his anonymous best reply is (we leave out the computational details):

\begin{eqnarray}\label{eq:apred}
f_{x,y,\alpha}(x_1,x_2)=\begin{cases} \frac{(1-\alpha)\frac{(1-y)}{(1-x)}\cdot x+\alpha\frac{x}{y}\cdot y}{(1-\alpha)\frac{1-y}{1-x}+\alpha\frac{x}{y}} &\text{ if }\{x_1,x_2\}=\{x,y\} \\
0 &\text{ if }x_1=0 \text{ or } x_2=0\\
1 &\text{ if }x_1=1 \text{ or } x_2=1.
\end{cases}
\end{eqnarray}

The resulting relative loss for the ignorant aggregator is therefore:
\begin{equation}\label{eq:rev}
L(f_{x,y,\alpha}, M_{x,y,\alpha})=2(x-y)^2\frac{(1-\alpha)\frac{(1-y)}{(1-x)}\cdot\alpha\frac{x}{y}}{(1-\alpha)\frac{(1-y)}{(1-x)}+\alpha\frac{x}{y}}.
\end{equation}

Given a pair of values, $(x,y)$, the corresponding $\alpha$ that maximizes the relative loss, $L(M_{x,y,\alpha},f_{x,y,\alpha})$, is denoted by  $\alpha^*(x,y)$. From the first-order conditions we derive a closed-form solution for $\alpha^*(x,y)$:

$$\alpha^*(x,y)=\frac{\sqrt{y(1-y)}}{\sqrt{y(1-y)}+\sqrt{x(1-x)}}.$$

%


An ignorant aggregator who does not know $\alpha$, but knows $x,y$ (because these are the two forecasts of the experts), may assume that the adversary chooses $\alpha=\alpha^*$. In such a case the ignorant aggregator's forecast
is obtained by replacing $\alpha$ with the formula for $\alpha^*(x,y)$ in the first case of Equation \eqref{eq:apred}:
\begin{align}\label{eq:far-point}
\p_{x,y,\alpha^*}(\omega=1|\{x_1,x_2\}=\{x,y\})= \frac{\big(\sqrt{y(1-y)}\big)x+\big(\sqrt{x(1-x)}\big)y}{\sqrt{y(1-y)}+\sqrt{x(1-x)}}\nonumber
\\
=\frac{\sqrt{\phi(x)}}{\sqrt{\phi(x)}+\sqrt{\phi(y)}}\cdot x + \frac{\sqrt{\phi(y)}}{\sqrt{\phi(x)}+\sqrt{\phi(y)}}\cdot y,
\end{align}
where $\phi$ is the precision function. Note that this coincides with the precision scheme whenever the two forecasts are sufficiently far apart.

Recall that we derive Equation \eqref{eq:far-point} under the assumption that the adversary chooses the value $\alpha=\alpha^*(x,y)$, which is optimal against an aggregator who best replies to a known martingale.
Once the ignorant aggregator fixes the scheme provided by Equation \eqref{eq:far-point}, we can reconsider the optimal martingale for the adversary. It turns out that in some cases the adversary can choose some value $\alpha \neq \alpha^*(x,y)$ and increase the relative loss. This happens only for martingales $\p_{x,y,\alpha}$ where $x-y < 0.4$. The interim conclusion was that the scheme provided in Equation \eqref{eq:far-point} guarantees the regret $\frac{1}{8}(5\sqrt 5-11)$ for all martingales $\p_{x,y,\alpha}$ where $x-y \geq 0.4$. By adjusting the weights to $(\frac{\phi(x)}{\phi(x)+\phi(y)},\frac{\phi(y)}{\phi(x)+\phi(y)})$, whenever $x-y<0.4$, we derive the optimal scheme.%
\footnote{The adjustment was inspired by simulations that demonstrated the reason for the failure of aggregation scheme \eqref{eq:far-point} for close forecasts.}
\section{More on conditionally independent signals}\label{app:ci}
\begin{proposition*}[Proposition \ref{pro:iid} restated]
For conditionally independent information structures, the average prior scheme guarantees a regret of $0.0250$.
\end{proposition*}
The proof of Proposition \ref{pro:iid} is very  similar to the proof of Theorem \ref{th:iid}.
\begin{proof}[Proof of Proposition \ref{pro:iid}]
We note that the arguments in the proof of Theorem \ref{th:iid} are universal for all aggregation schemes (not particularly the average prior scheme). The only change is in the definition of $r(\mu,x_1,x_2)$ (Equation \eqref{eq:r}), where $r$ in our case is given by
$$r(\mu,x_1,x_2)=\left(g(\mu,x_1,x_2)-g(ep(x_1,x_2),x_1,x_2)\right)^2.$$
Matlab maximization of the five-dimensional objective function $R(f,(\mu,\p^1_{\mu,y^1,z^1},\p^2_{\mu,y^2,z^2}))$ (see Equation \eqref{eq:5dObj}) yields a global maximum of $0.0250$ at the point $\mu=0.114$, $y_1=z_1=0.114$, $y_2=0$, $z_2=0.744$.
\end{proof}
Now we turn to some intuition about Conjecture \ref{con:iid}.
\begin{conj*}[Conjecture \ref{con:iid} restated]
For conditionally i.i.d. information structures, the minimal regret that can be guranteed is equal to $\frac{1}{8}(5\sqrt 5-11)\approx 0.0225$, and the average prior scheme gurantees this regret.
\end{conj*}
The proof of Theorem \ref{th:iid} shows that $\frac{1}{8}(5\sqrt 5-11)$ is a lower bound on the regret, even in identically distributed information structures, simply because the presented information structure is identically distributed. To deduce a lower bound, we note that indeed the maximum of the objective function $R(f,(\mu,\p^1_{\mu,y^1,z^1},\p^1_{\mu,y^1,z^1}))$ (see Equation \eqref{eq:5dObj}) in the identical case is equal to $\frac{1}{8}(5\sqrt 5-11)$. The only argument in the proof of Theorem \ref{th:iid} that fails for identical distributions is the fact that Equation \eqref{eq:lina},
\begin{align*}
R(f,(\mu,\p^1,\p^2))=\sum_{j,k} \alpha^1_j \alpha^2_k R(f,(\mu,\p^1_{\mu,y^1_j,z^1_j},\p^2_{\mu,y^2_k,z^2_k})).
\end{align*}
is no longer linear in the case where we add the identity restriction $\alpha^1=\alpha^2$, but rather turns out to be quadratic. We have failed to prove that the maximum of $R(f,(\mu,\p^1,\p^1))$ is obtained in support-two distribution $\p^1=\p^1_{\mu,x_1,x_2}$ (which will suffice to prove the conjecture); however, numerical simulations support this latter conjecture.


\section{Proof of Proposition \ref{pro:abag}}\label{sec:abag-proof}
It is sufficient to show a mixed strategy of the adversary, where an ignorant aggregator (who best replies to this strategy) cannot achieve a square loss smaller than $\frac{1}{4}-\epsilon$ whereas an ``almost Bayesian" aggregator has a square loss of $0$; that is, he always knows the state.

We define the following two information structures $\p_1(\delta),\p_2(\delta)$:

\begin{table}[h]
\centering
$\p_1(\delta):$
\begin{tabular}{cccccccc}
\multirow{4}{*}{} & \multicolumn{3}{c}{$\omega=0$}                                                    &  & \multicolumn{3}{c}{$\omega=1$}                                                    \\
                          &                             & $s_2$                    & $s'_2$                   &  &                             & $s_2$                    & $s'_2$                   \\ \cline{3-4} \cline{7-8}
                          & \multicolumn{1}{c|}{$s_1$}  & \multicolumn{1}{c|}{1/4} & \multicolumn{1}{c|}{0}   &  & \multicolumn{1}{c|}{$s_1$}  & \multicolumn{1}{c|}{0}   & \multicolumn{1}{c|}{$(1+\delta)/4$} \\ \cline{3-4} \cline{7-8}
                          & \multicolumn{1}{c|}{$s'_1$} & \multicolumn{1}{c|}{0}   & \multicolumn{1}{c|}{1/4} &  & \multicolumn{1}{c|}{$s'_1$} & \multicolumn{1}{c|}{$(1-\delta)/4$} & \multicolumn{1}{c|}{0}   \\ \cline{3-4} \cline{7-8}
\end{tabular}

\vspace*{4mm}
$\p_2(\delta):$
\begin{tabular}{ccccccc}
\multicolumn{3}{c}{$\omega = 0$}                                                                                                &  & \multicolumn{3}{c}{$\omega = 1$}                                                                                                    \\
                            & $s_2$                                           & $s'_2$                                          &  &                             & $s_2$                                             & $s'_2$                                            \\ \cline{2-3} \cline{6-7}
\multicolumn{1}{c|}{$s_1$}  & \multicolumn{1}{c|}{$(1+\delta)/(4-2\delta^2)$} & \multicolumn{1}{c|}{0}                          &  & \multicolumn{1}{c|}{$s_1$}  & \multicolumn{1}{c|}{0}                            & \multicolumn{1}{c|}{$(1-\delta^2)/(4-2\delta^2)$} \\ \cline{2-3} \cline{6-7}
\multicolumn{1}{c|}{$s'_1$} & \multicolumn{1}{c|}{0}                          & \multicolumn{1}{c|}{$(1-\delta)/(4-2\delta^2)$} &  & \multicolumn{1}{c|}{$s'_1$} & \multicolumn{1}{c|}{$(1-\delta^2)/(4-2\delta^2)$} & \multicolumn{1}{c|}{0}                            \\ \cline{2-3} \cline{6-7}
\end{tabular}
\end{table}

In $\p_1(\delta)$ the induced posteriors of the experts are as follows:
\begin{align*}
(\omega=0,s_1,s_2)) &\rightsquigarrow (\frac{1+\delta}{2+\delta},\frac{1-\delta}{2-\delta})\\
(\omega=0,s'_1,s'_2)) &\rightsquigarrow (\frac{1-\delta}{2-\delta},\frac{1+\delta}{2+\delta})\\
(\omega=1,s_1,s'_2)) &\rightsquigarrow (\frac{1+\delta}{2+\delta},\frac{1+\delta}{2+\delta})\\
(\omega=1,s'_1,s_2)) &\rightsquigarrow (\frac{1-\delta}{2-\delta},\frac{1-\delta}{2-\delta}).
\end{align*}

In $\p_2(\delta)$ the induced posteriors of the experts are as follows:
\begin{align*}
(\omega=0,s_1,s_2)) &\rightsquigarrow (\frac{1-\delta}{2-\delta},\frac{1-\delta}{2-\delta})\\
(\omega=0,s'_1,s'_2)) &\rightsquigarrow (\frac{1+\delta}{2+\delta},\frac{1+\delta}{2+\delta})\\
(\omega=1,s_1,s'_2)) &\rightsquigarrow (\frac{1-\delta}{2-\delta},\frac{1+\delta}{2+\delta})\\
(\omega=1,s'_1,s_2)) &\rightsquigarrow (\frac{1+\delta}{2+\delta},\frac{1-\delta}{2-\delta}).
\end{align*}

Assume that the adversary chooses the information structures $\p_1(\delta),\p_2(\delta)$ with equal probability $\frac{1}{2}$. An ``almost Bayesian" aggregator knows the actual information structure ($\p_1(\delta)$ or $\p_2(\delta)$), and can deduce from the pair of forecasts the state: in $\p_1(\delta)$, if forecasts are identical the state is $\omega=1$; otherwise the state is $\omega=0$. In $\p_2(\delta)$ the opposite policy holds.

Let $\sigma_1$ be the event that the information structure is $\p_1(\delta)$, the state is $\omega=0,$ and the signals are $s_1,s_2$. Let $\sigma_2$ be the event that the information structure is $\p_2(\delta)$, the state is $\omega=1$, and the signals are $s_1,s'_2$. An ignorant aggregator who knows only the mixed strategy of the adversary, when he observes the pair of forecasts $(\frac{1+\delta}{2+\delta},\frac{1-\delta}{2-\delta}),$ assigns the probabilities $$(p_1,p_2)=(\frac{\frac{1}{4}}{\frac{1}{4}+\frac{1-\delta^2}{4-2\delta^2}},
\frac{\frac{1-\delta^2}{4-2\delta^2}}{\frac{1}{4}+\frac{1-\delta^2}{4-2\delta^2}})=(\frac{2-\delta^2}{4-3\delta^2},\frac{2-2\delta^2}{4-3\delta^2})$$
to the events $(\sigma_1,\sigma_2)$ respectively. Therefore his optimal forecast is $\frac{2-2\delta^2}{4-3\delta^2}=\frac{1}{2}\pm O(\delta)$.

Similar considerations can be applied to the cases where the pair of observed forecasts are $(\frac{1-\delta}{2-\delta},\frac{1+\delta}{2+\delta}), (\frac{1-\delta}{2-\delta},\frac{1+\delta}{2+\delta}),$ and $(\frac{1+\delta}{2+\delta},\frac{1-\delta}{2-\delta})$. In all these cases we get that the optimal forecast is $\frac{1}{2}\pm O(\delta)$. Therefore the square loss of the ignorant aggregator is at least $(\frac{1}{2}-O(\delta))^2=\frac{1}{4}-O(\delta)$, and we can set small enough $\delta$ such that $O(\delta)$ in the last expression is less than $\epsilon$.

\section{Proof of Theorem \ref{th:freg}}\label{ap:reg-proof}
We construct a mixed strategy for the adversary $\sigma=\sigma(k)$ with the following two properties that are stated as lemmas.
\begin{lemma}\label{lem:r-bound}
The expected square loss in $\sigma(k)$ of a Bayesian aggregator who knows $\sigma(k)$ (but does not know the realized information structure) is at least $\frac{1}{4} - \frac{1}{3k}$.
\end{lemma}

\begin{lemma}\label{lem:ub}
The omniscient expert incurs an expected square loss in $\sigma(k)$ of at most $e^{-\frac{n}{72k^2}}$.
\end{lemma}
Note that these two lemmas complete the proof of the theorem, because for any aggregation scheme the expected loss in $\sigma(k)$ is at least $\frac{1}{4} - \frac{1}{3k}$ (because Bayesian aggregation is the optimal aggregation). Therefore, the expected relative loss is at least $\frac{1}{4} - \frac{1}{3k} - e^{-\frac{n}{72k^2}}$ in $\sigma(k)$. Therefore there exists an information structure (in the support of $\sigma(k)$) for which the relative loss is at least $\frac{1}{4} - \frac{1}{3k} - e^{-\frac{n}{72k^2}}$. We set $k=\sqrt{\frac{n}{72 \log n}}$, and obtain that the regret is at least $$R_{\mathcal{D}_n} \geq \left( \frac{1}{4} - \frac{1}{3}\sqrt{\frac{72 \log n}{n}} \right) - \frac{1}{n} \geq \frac{1}{4} - 3 \sqrt{\frac{\log n}{n}}.$$

We first present the mixed strategy of the adversary (i.e., a distribution over information structures) $\sigma(k)$ along with some preliminaries that will be useful in the proofs of the lemmas.

We fix $k$, and we consider the sequence $y_1,...,y_k$ defined by
\begin{align}\label{eq:yym}
y_{m}=\frac{1}{1+( 1- \frac{2}{2k+1} )^{2m-2}} \text{ which is bounded by } \frac{1}{2} \leq  y_{m} \leq \frac{1}{1+e^{-2}} < \frac{9}{10}.
\end{align}
In Table \ref{tb:i} we define the information structure \emph{for a single expert} $I(m)$  for $1\leq m\leq k$ (i.e., a correlated distribution over states and signals).
\begin{table}[h]
\centering
\caption{The information structure $I(m)$.}
\label{tb:i}
\begin{tabular}{ccc}
                           & $\omega=0$                         & $\omega=1$                         \\ \cline{2-3}
\multicolumn{1}{c|}{$s_0$} & \multicolumn{1}{c|}{$(\frac{1}{2}+\frac{1}{4k}) (1-y_{m})$} & \multicolumn{1}{c|}{$(\frac{1}{2}-\frac{1}{4k}) (1-y_{m})$} \\ \cline{2-3}
\multicolumn{1}{c|}{$s_1$} & \multicolumn{1}{c|}{$(\frac{1}{2}-\frac{1}{4k}) y_{m}$} & \multicolumn{1}{c|}{$(\frac{1}{2}+\frac{1}{4k}) y_{m}$} \\ \cline{2-3}
\end{tabular}
\end{table}
We denote the prior of $I(m)$ by
\begin{align}\label{eq:pri-bound}
\mu_{m}=\frac{1}{2}-\frac{1}{4k}+\frac{1}{2k}y_{m} \text{ which is bounded by } \frac{1}{2}\leq \mu_{m} \leq \frac{1}{2} + \frac{1}{5k}.
\end{align}
We set $\beta_1=1$, and we recursively define $\beta_2,...,\beta_k$ by
$\beta_{m+1}=\beta_{m} \frac{\mu_{m}}{1-\mu_{m+1}}$.
Equation \eqref{eq:pri-bound} implies that $$\beta^i\leq \beta^{i+1}= \frac{\mu^i}{1-\mu^{i+1}} \beta^{i} \leq \frac{\frac{1}{2}+\frac{1}{5k}}{\frac{1}{2}} \beta^{i}= (1 + \frac{2}{5k}) \beta^{i}.$$
Therefore,
\begin{align}\label{eq:beta}
1=\beta^1 \leq \beta^i \leq (1+\frac{2}{5k})^{i-1} \beta^1 \leq (1+\frac{2}{5k})^{k-1} \leq e^{0.4} < 1.5.
\end{align}
We normalize $(\beta_i)$ to a probability distribution by setting $\alpha_m=\frac{\beta_m}{\sum_i \beta_i}$. Note that $(\alpha_i)$ satisfies
\begin{align}\label{eq:alpha}
\alpha_m \mu_m = \alpha_{m+1}\mu_{m+1}.
\end{align}
The distribution $\sigma(k)$ assigns a probability of $\alpha_{m}$ to the information structure $I(m)$ (i.e., it randomizes over $k$ information structures). Now we prove the two lemmas.

\begin{proof}[Proof of Lemma \ref{lem:r-bound}]
By the definition of $I(m)$ (see Table \ref{tb:i}), an expert who observes $s_0$ forecasts $\frac{1}{2}-\frac{1}{4k}$, and an expert who observes $s_1$ forecasts $\frac{1}{2}+\frac{1}{4k}$. Therefore, in state $\omega=0$ the induced distribution over forecasts is
\begin{align}
I_0(m)=\begin{cases} \frac{1}{2}-\frac{1}{4k} \text{ with probability } \frac{1}{1-\mu_{m}}(\frac{1}{2}+\frac{1}{4k}) (1-y_{m}) \\
\frac{1}{2}+\frac{1}{4k} \text{ with probability } \frac{1}{1-\mu_{m}}(\frac{1}{2}-\frac{1}{4k})y_{m},
\end{cases}
\end{align}
and similarly
\begin{align}
I_1(m)=\begin{cases} \frac{1}{2}-\frac{1}{4k} \text{ with probability } \frac{1}{\mu_{m}}(\frac{1}{2}-\frac{1}{4k}) (1-y_{m}) \\
\frac{1}{2}+\frac{1}{4k} \text{ with probability } \frac{1}{\mu_{m}}(\frac{1}{2}+\frac{1}{4k})y_{m}.
\end{cases}
\end{align}
The Bayesian aggregator (who knows $\sigma$) observes a sample of size $n$ from the distribution $I_\omega(m)$ in the case where the realized (by $\sigma$) information structure is $I(m)$ and the realized (by nature) state is $\omega$. This Bayesian aggregtor suffers an expected square loss at least as high as a Bayesian aggregator (who knows $\sigma$) and observes the distribution $I_\omega(m)$ \emph{precisely} rather than just a sample from it. This follows from the fact that observing the distribution Blackwell dominates the information of a finite sample from it. Therefore it is sufficient to prove that a Bayesian aggregator (who knows $\sigma$) and observes the distribution $I_\omega(m)$ precisely incurs a square loss of at least $\frac{1}{4}-\frac{1}{k}$.

The key property of the sequence $(I(m))_{m=1,...,k}$ is that $I_1(m)=I_0(m+1)$ for $m=1,...,k-1$. To see this, we consider the likelihoods of the forecasts in the two distributions:
\begin{align*}
\frac{(\frac{1}{2}+\frac{1}{4k})y_{m}}{(\frac{1}{2}-\frac{1}{4k}) (1-y_{m})} &= \frac{(\frac{1}{2}-\frac{1}{4k})y_{m+1}}{(\frac{1}{2}+\frac{1}{4k}) (1-y_{m+1})} \Leftrightarrow \\
\frac{\frac{1}{2}+\frac{1}{4k}}{\frac{1}{2}-\frac{1}{4k}} \left(1-\frac{2}{2k+1}\right)^{-(2m-2)}&= \frac{\frac{1}{2}-\frac{1}{4k}}{\frac{1}{2}+\frac{1}{4k}} \left(1-\frac{2}{2k+1}\right)^{-2m}\Leftrightarrow \\
\left(1-\frac{2}{2k+1}\right)^{2} &= \left( \frac{\frac{1}{2}-\frac{1}{4k}}{\frac{1}{2}+\frac{1}{4k}} \right)^2
\end{align*}
By the definition of $(\alpha_i)$ (Equation \eqref{eq:alpha}), $\alpha_m \mu_m = \alpha_{m+1} (1-\mu_{m+1})$. Therefore,
a Bayesian aggregator who observes a distribution of forecasts $J=I_1(m)=I_0(m+1)$ assigns an equal probability of $\frac{1}{2}$ to the two events that the information structure is $I(m)$ and the realized state is $\omega=1$, on the one hand, and that the information structure is $I(m+1)$ and the realized state is $\omega=0$, on the other. Therefore, in either such event the Bayesian aggregator will forecast $\frac{1}{2}$, and his square loss will be $\frac{1}{4}$. Using inequalities \eqref{eq:pri-bound} and \eqref{eq:beta} we get the following bound on the aggregator's square loss:
\begin{align*}
\frac{1}{4}(1-\alpha_1(1-\mu_1)-\mu_k \alpha_k) &\geq \frac{1}{4}(1-\frac{1}{2}\frac{1}{\sum_i \beta_i}-0.51\frac{\beta_k}{\sum_i \beta_i}) \\ &\geq
\frac{1}{4}(1-\frac{1}{2k}-0.51\frac{1.5}{k}) \geq \frac{1}{4}-\frac{1}{3k}.
\end{align*}
\end{proof}
\begin{proof}[Proof of Lemma \ref{lem:ub}]
According to the information structure $I(m)$, conditional on state $\omega=1$, the expected number of experts whose posterior is $\frac{1}{2}+\frac{1}{4k}$ equals
\begin{align}\label{eq:1f}
\frac{(\frac{1}{2}+\frac{1}{4k}) y_m }{ (\frac{1}{2}+\frac{1}{4k}) y_m + (\frac{1}{2}-\frac{1}{4k}) (1-y_m)}.
\end{align}
Similarly, conditional on $\omega=0$, the expected number of experts whose posterior is $\frac{1}{2}+\frac{1}{4k}$ is
\begin{align}\label{eq:0f}
\frac{(\frac{1}{2}-\frac{1}{4k}) y_m }  {(\frac{1}{2}-\frac{1}{4k}) y_m + (\frac{1}{2}+\frac{1}{4k}) (1-y_m)}.
\end{align}
We denote by $D(k,m)$ the difference in these expectations:
$$D(k,m)=\frac{(\frac{1}{2}+\frac{1}{4k}) y_m }{ (\frac{1}{2}+\frac{1}{4k}) y_m + (\frac{1}{2}-\frac{1}{4k}) (1-y_m)} -
\frac{(\frac{1}{2}-\frac{1}{4k}) y_m }  {(\frac{1}{2}-\frac{1}{4k}) y_m + (\frac{1}{2}+\frac{1}{4k}) (1-y_m)}.$$
Note that the function
$$D(y)=\frac{(\frac{1}{2}+\frac{1}{4k}) y }{ (\frac{1}{2}+\frac{1}{4k}) y + (\frac{1}{2}-\frac{1}{4k}) (1-y)} -
\frac{(\frac{1}{2}-\frac{1}{4k}) y }  {(\frac{1}{2}-\frac{1}{4k}) y + (\frac{1}{2}+\frac{1}{4k}) (1-y)}.$$
is monotonically decreasing in $y\in [\frac{1}{2},1]$ because the derivative
$$\frac{D(y)}{dy}=-\frac{
(\frac{1}{2}+\frac{1}{4k})
(\frac{1}{2}-\frac{1}{4k})
\frac{1}{2k}
(2y-1)
}{
(\frac{1}{2}-\frac{1}{4k}+\frac{1}{2k}y)^2
(\frac{1}{2}+\frac{1}{4k}-\frac{1}{2k}y)^2
}$$
is negative for $y>\frac{1}{2}$. Now we are able to bound the expression
$D(k,m)$,
\begin{align}\label{eq:Db}
\begin{split}
D(k,m)&\geq
\frac{(\frac{1}{2}+\frac{1}{4k}) \frac{9}{10} }{ (\frac{1}{2}+\frac{1}{4k}) \frac{9}{10} + (\frac{1}{2}-\frac{1}{4k}) \frac{1}{10}} -
\frac{(\frac{1}{2}-\frac{1}{4k}) \frac{9}{10} }  {(\frac{1}{2}-\frac{1}{4k}) \frac{9}{10} + (\frac{1}{2}+\frac{1}{4k}) \frac{1}{10}} \\
&=\frac{18k+9}{20k+8}-\frac{18k-9}{20k-8}=\frac{9k}{50k^2-8}\geq \frac{1}{6k},
\end{split}
\end{align}
when the first inequality follows from inequality \eqref{eq:yym} and the monotonicity of $D(y)$.

Now we turn to the proof of the lemma. For every distribution $I(m)$ in the support of $\sigma(k)$ we introduce an aggregation scheme (which depends on $I(m)$) that guarantees a square loss of at most $e^{-\frac{n}{72k^2}}$. This obviously implies that the omniscient expert incurs at most the same loss.
Our aggregation scheme has the following simple form: it counts the fraction $q$ of experts whose forecast is $\frac{1}{2}-\frac{1}{4k}$. If $q$ is closer to the expected fraction in state $1$ (Equation \eqref{eq:1f}) then it forecasts $1$. Otherwise, if $q$ is closer to the expected fraction in state $0$ (Equation \eqref{eq:1f}) then it forecasts $0$. It never uses probabilistic forecasts.
For this simple $0/1$ aggregation scheme the expected square loss is equal to the probability of a mistake. By inequality \eqref{eq:Db} a mistake in the prediction occurs only in the case where the fraction $q$ is $\frac{1}{12k}$-far from the expected fraction. By Hoeffding's inequality the probability of this event is at most $e^{-\frac{n}{72k^2}}$, which is also a bound on the expected square loss of the presented aggregation scheme.
\end{proof}

\newpage

\section{Contradictory weather forecasts}\label{app_weather}
\begin{figure}[h]
\begin{center}
\caption{Forecast in Yahoo! website}\label{pic1}
\includegraphics[scale=1]{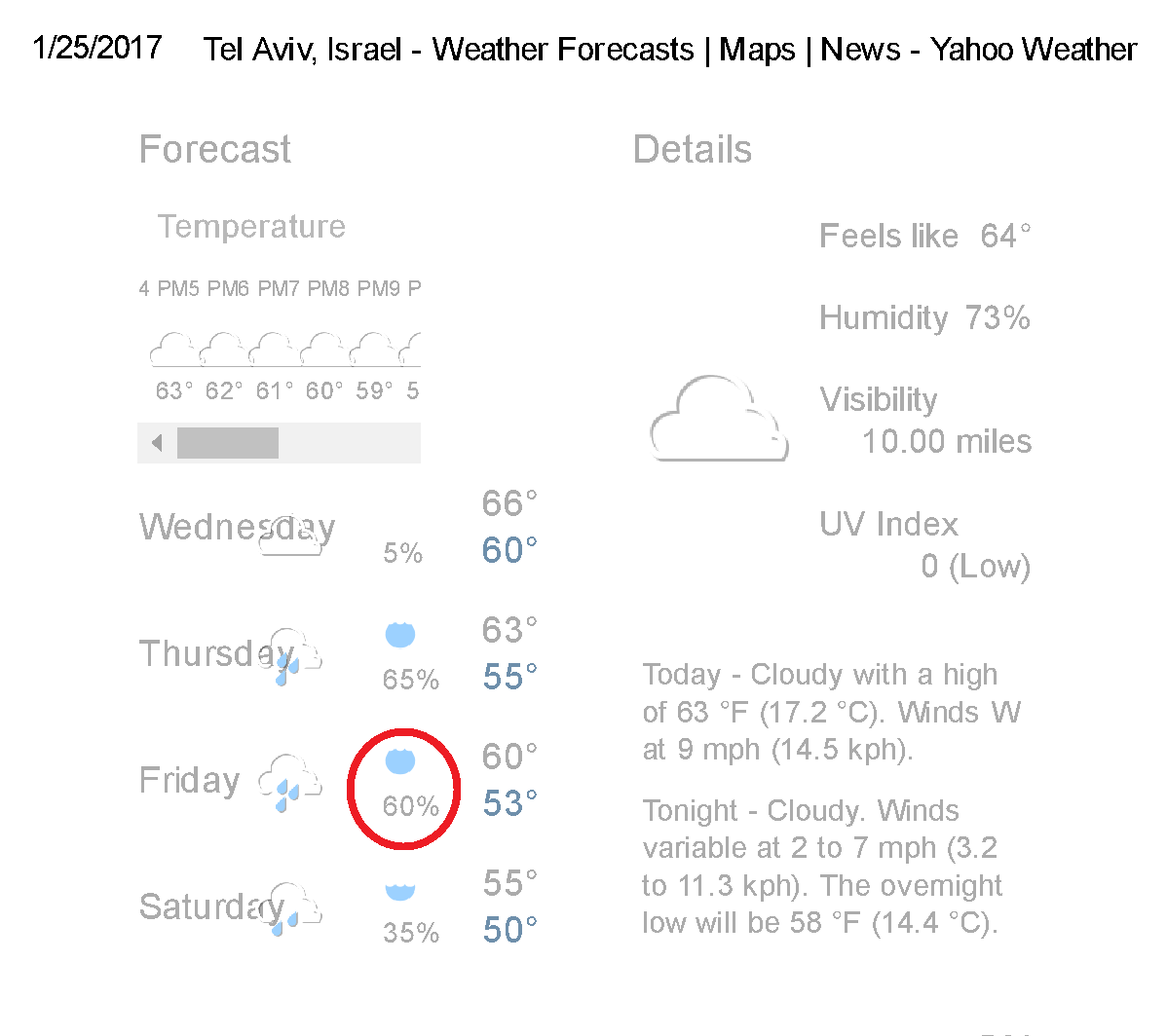}
\end{center}
\end{figure}

\begin{figure}[h]
\begin{center}
\caption{Forecast in Accuweather website}\label{pic2}
\includegraphics[scale=1]{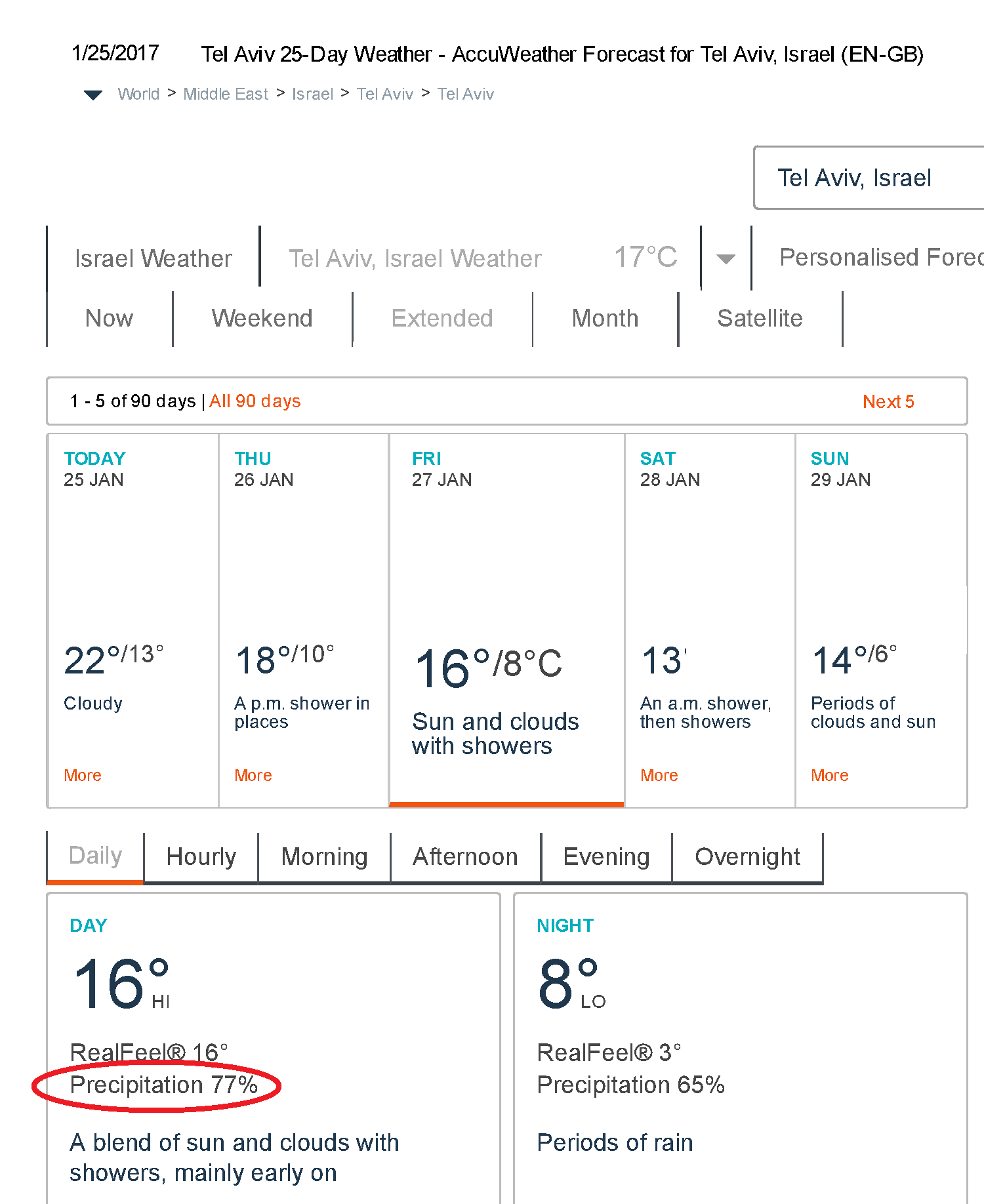}
\end{center}
\end{figure}

\begin{figure}[h]
\begin{center}
\caption{Forecast in Weather-Channel website}\label{pic3}
\includegraphics[scale=0.6]{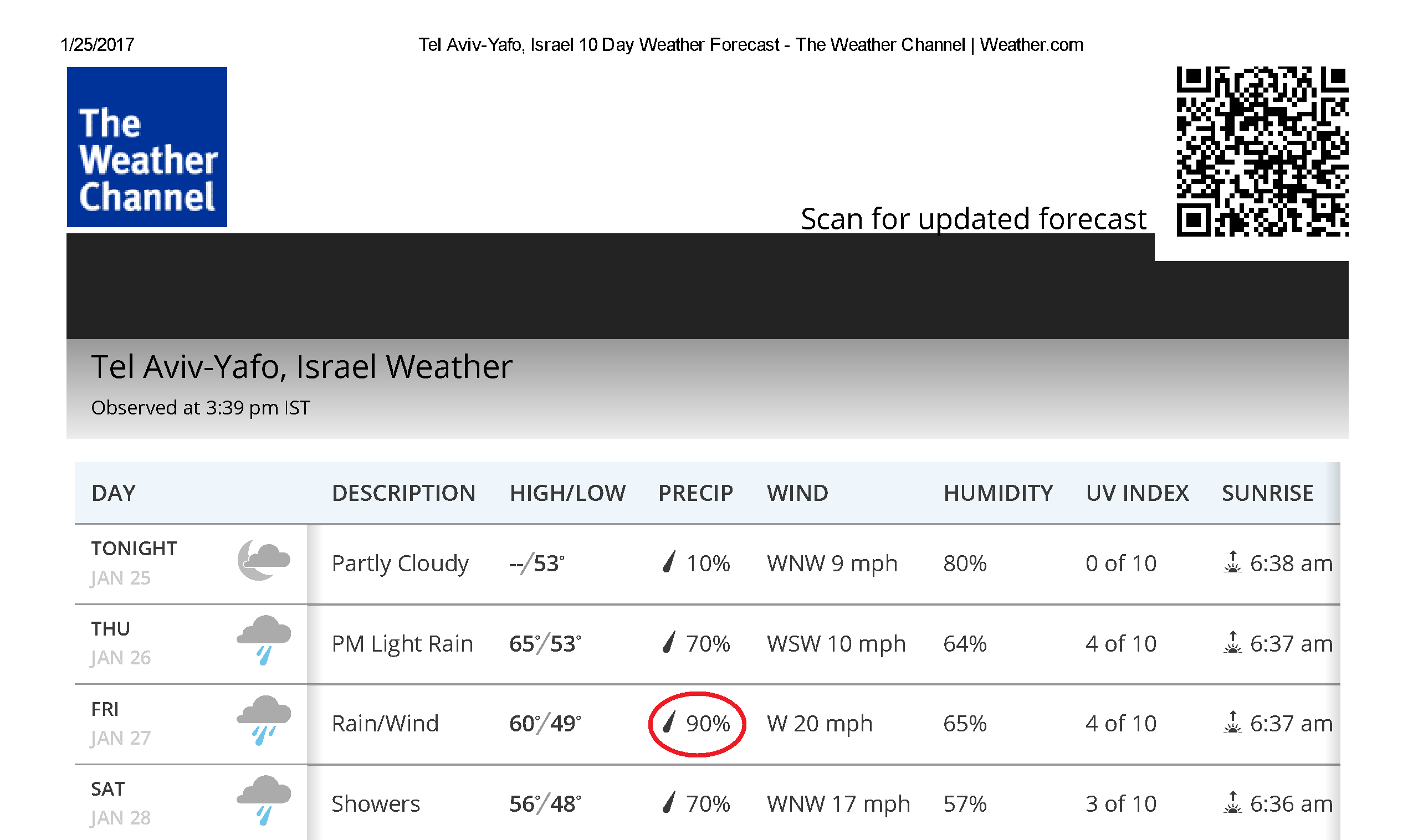}
\end{center}
\end{figure}

\end{document}